\documentclass[a4paper,12pt,reqno]{amsart}
\usepackage{amsmath,amssymb,latexsym,amsthm,mathrsfs,xcolor}
\usepackage[inner=2.54cm,outer=2.54cm,top=3.17cm,bottom=3.17cm]{geometry}
\usepackage[colorlinks,linkcolor=blue,citecolor=red]{hyperref}
\theoremstyle{definition}
\newtheorem{Theorem}{Theorem}[section]
\newtheorem{Lemma}[Theorem]{Lemma}
\newtheorem{Crollary}[Theorem]{Corollary}
\newtheorem{Remark}[Theorem]{Remark}
\newtheorem{Proposition}[Theorem]{Proposition}
\newtheorem*{Hypothesis1}{Hypothesis H1}
\newtheorem*{Hypothesis2}{Hypothesis H2}
\numberwithin{equation}{section}
\newcommand{\mi}{\mathrm{i}}
\newcommand{\md}{\mathrm{d}}
\newcommand{\me}{\mathrm{e}}
\newcommand{\diag}{\mathrm{diag}}
\newcommand{\meas}{\mathrm{Meas}}
\newcommand{\id}{\mathrm{Id}}
\newcommand{\mth}{\mathrm{th}}

\def\Z{{\mathbb Z}}
\def\R{{\mathbb R}}
\def\T{{\mathbb T}}

\title[Reducibility of 1-D Quantum Harmonic Oscillator]
{Reducibility of 1-D Quantum Harmonic Oscillator with  Decaying Conditions on the Derivative of Perturbation Potentials}

\author{Zhenguo Liang}
\address{School of Mathematical Sciences and Key Lab of Mathematics for Nonlinear Science, Fudan University, Shanghai 200433, China}
\email{zgliang@fudan.edu.cn}
\thanks{Z. Liang was partially supported by National Natural Science Foundation of China (Grants No. 12071083)
and Natural Science Foundation of Shanghai (Grants No. 19ZR1402400)}

\author{Zhiqiang Wang}
\address{School of Mathematical Sciences and Key Lab of Mathematics for Nonlinear Science, Fudan University, Shanghai 200433, China}
\email{19110180010@fudan.edu.cn}

\begin{document}
\begin{abstract}
We prove  the reducibility of 1-D quantum harmonic oscillators in $\mathbb R$ perturbed by a quasi-periodic in time potential $V(x,\omega t)$ under the following conditions, namely  
 there is a $C>0$ such that
\begin{equation*}
|V(x,\theta)|\le C,\quad|x\partial_xV(x,\theta)|\le C,\quad\forall~(x,\theta)\in\mathbb R\times\mathbb T_\sigma^n.
\end{equation*}
The corresponding perturbation matrix $(P_i^j(\theta))$ is proved to satisfy 
 $
 (1+|i-j|)| P_i^j(\theta)|\le C$ and  $\sqrt{ij}|P_{i+1}^{j+1}(\theta)-P_i^j(\theta)|\le  C$
 for any $\theta\in\mathbb T_\sigma^n$ and $i,j\geq 1$. A new reducibility theorem is set  up under this kind of decay in the perturbation matrix element $P_{i}^j(\theta)$ as well as  the discrete difference matrix element
 $P_{i+1}^{j+1}(\theta)-P_i^j(\theta)$. For the proof the novelty is that we use the decay in the discrete difference matrix element to control  the measure estimates for the thrown parameter sets. 
\end{abstract}
\maketitle
\section{Introduction}
\subsection{The Statement of Main Result}
In this paper we consider the linear Schr\"odinger equation
\begin{equation}\label{Schrodinger}
\mi\partial_t\psi=(-\partial_x^2+x^2)\psi+\epsilon V(x,\omega t)\psi,~\psi=\psi(x,t),~x\in\mathbb R,
\end{equation}
where $\epsilon\ge0$ is a small parameter and frequency vector $\omega$ of the forced oscillations is regarded as a parameter in $\mathcal D_0:=[0,2\pi)^n$. We assume  that the potential  $V:\mathbb R\times\mathbb T^n\ni(x,\theta)\mapsto V(x,\theta)\in\mathbb R$ is $\mathcal C^1$ smooth and analytic in $\theta\in\mathbb T^n$, where $\mathbb T^n=\mathbb R^n/2\pi\mathbb Z^n$ denotes the $n$ dimensional torus. For $\sigma>0$, the function $V(x,\theta)$ analytically in $\theta$ extends to the strip $\mathbb T^n_\sigma=\{a+b\mi\in\mathbb C^n/2\pi\mathbb Z^n:|b|<\sigma\}$ and is bounded on $\mathbb R\times\mathbb T^n_\sigma$ as well as $x\partial_xV$, namely there is a $C>0$ such that
\begin{equation}\label{PotentialAssumption}
|V(x,\theta)|\le C,\quad|x\partial_xV(x,\theta)|\le C,\quad\forall~(x,\theta)\in\mathbb R\times\mathbb T_\sigma^n.
\end{equation}\par
To state our main results we first introduce some notations. Let $H=-\partial_x^2+x^2$ be the Hermite operator. It is well-known that $H$ has a simple pure point spectrum $\nu_i=2i-1$ so that  $Hh_i(x)=\nu_ih_i(x)$, where $h_i(x)$ is the so-called Hermite function and satisfies $\|h_i\|_{L^2}=1$ for $i\ge1$ . Clearly, $(h_i)_{i\ge1}$ forms an orthonormal basis of $L^2(\mathbb R)$, called the Hermite basis. For $s\ge0$, denote by $\mathcal H^s=\mathbb D(H^\frac s2)=\{f\in L^2(\mathbb R):H^\frac s2 f\in L^2(\mathbb R)\}$ the domain of $H^\frac s2$ endowed by the graph norm. For negative $s$, the space $\mathcal H^s$ is the dual of $\mathcal H^{-s}$. For Hilbert spaces $\mathcal H_1$ and $\mathcal H_2$, we will denote by $\mathcal B(\mathcal H_1,\mathcal H_2)$ the space of bounded linear operators from $\mathcal H_1$ to $\mathcal H_2$ and write $\mathcal B(\mathcal H_1,\mathcal H_1)$ as $\mathcal B(\mathcal H_1)$ for simplicity. \\
\indent Let $s\in\mathbb R$, we define the complex weighted $\ell^2$- space
$$
\ell_s^2:=\{\xi=(\xi_i\in\mathbb C,i\in\mathbb N_0):\|\xi\|_s<\infty\}\text{ with }\|\xi\|_s^2=\sum_{i\in\mathbb N_0}i^s|\xi_i|^2.
$$
Notice for later use that to a function $f\in\mathcal H^s$ we formulate $f(x)=\sum_{i\ge1}\xi_ih_i(x)$, where the sequence $\xi=(\xi_i)_{i\ge1}\in \ell_s^2$ is the so-called Hermite coefficients.
 We will identify $\mathcal H^s$ with $\ell_s^2$ by endowing both spaces with the norm $\|f\|_{\mathcal H^s}=\|\xi\|_s=\left(\sum_{i\ge1}i^s|\xi_i|^2\right)^\frac12$. Then our main results are the following.
\begin{Theorem}\label{MainTheorem}
Assume that a $\mathcal C^1$ smooth function $V(x,\theta)$  analytically in $\theta$ extends to $\mathbb T_\sigma^n$ for $\sigma>0$ and fulfills \eqref{PotentialAssumption}. There exists $\epsilon_*>0$ such that for all $0\le\epsilon<\epsilon_*$ there is a subset $\mathcal D_\epsilon\subset\mathcal D_0=[0,2\pi)^n$ of asymptotically full measure such that for all $\omega\in\mathcal D_\epsilon$, the linear Schr\"odinger equation \eqref{Schrodinger} reduces to a linear equation with constant coefficients (w.r.t. time variable) in $L^2(\mathbb R)$.\par
More precisely, let $p\in[0,3)$ then for any $\omega\in\mathcal D_\epsilon$ there exists a linear isomorphism $\Psi_{\omega,\epsilon}(\theta)\in\mathcal{B}(\mathcal{H}^p)$, unitary on $L^2$, which analytically depends on $\theta\in\mathbb{T}^n_{\sigma/2}$ such that $t\mapsto\psi(\cdot,t)\in\mathcal{H}^p$ satisfies \eqref{Schrodinger} if and only if $t\mapsto\zeta(\cdot,t)=\Psi_{\omega,\epsilon}^{-1}(\omega t)\psi(\cdot,t)\in\mathcal{H}^p$ satisfies the autonomous equation $\mi\partial_t\zeta=H^\infty\zeta$ with $H^\infty=\diag(\lambda_i^\infty)_{i\ge1}$. Furthermore, there is a $C>0$ such that
\begin{equation*}
\begin{gathered}
\meas(\mathcal{D}_0\setminus\mathcal{D}_\epsilon)\le C\epsilon^{\frac1{51}},\\
|\lambda_i^\infty-\nu_i|\le C\epsilon,\quad\forall~i\ge1,\\
\|\Psi_{\omega,\epsilon}^{\pm1}(\theta)-\id\|_{\mathcal{B}(\mathcal{H}^p)}\le C\epsilon^\frac23,\quad\forall~\theta\in\mathbb{T}^n_{\sigma/2}.
\end{gathered}
\end{equation*}
\end{Theorem}
\begin{Remark} 
We give for example $V(x,\theta) := g(\theta) F\big(\langle x\rangle^{-\mu}\big)$ to fulfill \eqref{PotentialAssumption}, where $\mu>0$, $\langle x\rangle : =\sqrt{1+x^2}$, $g(\theta)$ is a real analytic  and bounded function on $\T_{\sigma}^n$ with $\sigma>0$.  The real function $F(\cdot )$  is $\mathcal C^1$ smooth  on $\R$ which satisfies 
$|F(\cdot )|, |F'(\cdot)|\leq C$ on $\R$ with some $C>0$.  
\end{Remark} 
Consequently, we have the following corollaries concerning the Sobolev norm estimations on the solution of \eqref{Schrodinger} and the spectra of the corresponding Floquet operator defined by $K_F:=-\mi\sum_{j=1}^n\omega_j\partial_{\theta_j}-\partial_x^2+x^2+\epsilon V(x,\theta)$.
\begin{Crollary}\label{SolutionEstimate}
Assume that a $\mathcal C^1$ smooth function $V(x,\theta)$ analytically in $\theta$ extends to $\mathbb T_\sigma^n$ for $\sigma>0$ and fulfills \eqref{PotentialAssumption}. There exists $\epsilon_*>0$ such that for all $0\le\epsilon<\epsilon_*$ and $\omega\in\mathcal{D}_\epsilon$, the linear Schr\"odinger equation \eqref{Schrodinger} has a unique solution $\psi(t)\in\mathcal C^0(\mathbb R,\mathcal H^p)$ subject to the initial datum $\psi(0)\in\mathcal H^p$ with $p\in[0,3)$. Moreover, $\psi$ is almost-periodic in time and satisfies for a $C>0$
\[
(1-C\epsilon)\|\psi(0)\|_{\mathcal H^p}\le\|\psi(t)\|_{\mathcal H^p}\le(1+C\epsilon)\|\psi(0)\|_{\mathcal H^p}.
\]
\end{Crollary}
\begin{Crollary}\label{PurePoint}
Assume that a $\mathcal C^1$ smooth function $V(x,\theta)$ analytically in $\theta$ extends to $\mathbb T_\sigma^n$ for $\sigma>0$ and fulfills \eqref{PotentialAssumption}. There exists $\epsilon_*>0$ such that for all $0\le\epsilon<\epsilon_*$ and $\omega\in\mathcal{D}_\epsilon$, the spectrum of the Floquet operator $K_F$ is pure point.
\end{Crollary}
\subsection{Discussions}
\indent The boundedness of Sobolev norms via reducibility, as well as the pure point nature of Floquet operator, was firstly considered for 1-D quantum harmonic oscillator (`QHO' for short) with time periodic smooth perturbations \cite{Com87, DLSV2002, EV83, Kuk1993}. As for time quasi-periodic perturbations, we can refer to \cite{GT2011, Wang08,WL2017}, in which the perturbation potentials are all supposed to be decaying to zero when $x$ goes to $\infty$.  \\
\indent Eliasson \cite{Eli2011} asked the following open question:
\textsl{``Is the 1-D quantum harmonic oscillators (\ref{Schrodinger}) reducible for the bounded perturbation i.e. $|V(x, \theta)|\leq C$, where $V(x, \theta )$ is smooth on $x$ and analytic on $\theta\in \mathbb T^n$?  ''
}This question follows  naturally from the results in \cite{BG01} and \cite{LY10} and keeps still unsolved till now except some special cases(see \cite{LiangLuo2021}).   
Theorem \ref{MainTheorem} gives an answer for Eliasson's question if we assume further   $|xV_x(x, \theta)|\leq C$.   \\
\indent For 1-D  QHO with  unbounded perturbations,  see Bambusi's recent papers \cite{Bam2017, Bam2018} based on the pseudo-differential calculus.  In fact Bambusi's  results in 1-D quantum harmonic oscillators \cite{Bam2018} can be translated into 
matrix language.  In \cite{Chodosh11} a discrete difference operator $\Delta$ on a function $K: \mathbb N_0\times \mathbb N_0\rightarrow \R$ is defined as 
$(\Delta K)_m^n : =\Delta K_{m}^n= K_{m+1}^{n+1}- K_m^n$,
where $\mathbb N_0 : =\{1,2,\cdots \}$. 
Write $\Delta^{M}$ to signify applying the difference operator $M$ 
times. Chodosh \cite{Chodosh11} called that a function $K: \mathbb N_0\times \mathbb N_0\rightarrow \R$ is a symbol matrix of order $r$ if for all $M$, $N\in \mathbb N_0$, there is $C_{M, N}>0$ such that 
$$|(\Delta^{M} K)_m^n|\leq C_{M,N} (1+m+n)^{r-M}(1+|m-n|)^{-N}. $$ 
Chodosh \cite{Chodosh11} proved that a operator $A$ with the symbol $a\in S^r$(\cite{Bam2018}, $l=1$)  if and only if the ``matrix of A'' 
$K^{(A)}: \mathbb N_0\times \mathbb N_0\rightarrow \R$
defined by $(m,n)\rightarrow \langle A h_n, h_m\rangle_{L^2(\R)}$ is an order of $r/2$ symbol matrix.   Bambusi's results \cite{Bam2018} together with \cite{Chodosh11} implied that  the discrete difference  matrix should play a very important role in the reducibility.  In the following we will give a heuristic proof how 
the discrete difference operator ``$\Delta$'' and the corresponding decay in the difference matrix element  work  in  the proof of Theorem \ref{MainTheorem}. \\
\indent For the following discussion we will introduce some notations.   Let $\alpha\in\mathbb R$, denote by $\mathcal M_\alpha$ the set of infinite matrices $A:\mathbb{N}_0\times\mathbb{N}_0\mapsto\mathbb C$ that satisfy $$|A|_\alpha:=\sup_{i,j\in\mathbb N_0}(1+|i-j|)^\alpha|A_i^j|<\infty.$$  For $\alpha, \beta\in\mathbb R$ we define
$\mathcal M_{\alpha,\beta}$,  the subset of $\mathcal M_\alpha$ as the following: an infinite matrix $A$ belongs to $\mathcal M_{\alpha,\beta}$ if 
$$|A|_{\alpha,\beta}:=\sup_{i,j\in\mathbb N_0}(1+|i-j|)^\alpha|A_i^j|+\sup_{i,j\in\mathbb N_0}(ij)^\beta|\Delta A_i^j|<\infty. $$ \par
 \indent Following  Bambusi and Graffi \cite{BG01}, we will prove  the reducibility of the equation
\[
{\rm i}\partial_t u(t)=\big(A_0+ P_0(\omega t)\big)u(t),  \qquad u\in \ell_0^2,
\]
where $\omega\in \mathcal{D}_0 : = [0,2\pi)^n$ and 
 $A_0=\diag(\nu_j)_{j\ge1}$ and   
$
\nu_j=j
$ for simplicity,
and the map $\mathbb T^n\ni\theta~\mapsto~P_0(\theta)\in \mathcal{M}_{\alpha, \beta}$ is analytic on $\mathbb T_\sigma^n$ with $0<\beta \leq \alpha$ and $(P_0)_i^j(\theta)=(P_0)_j^i(\theta)$ for $\theta\in\mathbb T^n$.  
Let us describe the iteration step. We begin from the equation 
$$
{\rm i}\partial_t u(t)=\big(A+ P(\omega t)\big)u(t),  \qquad u\in \ell_0^2, \qquad P(\theta)\in \mathcal{M}_{\alpha, \beta}. 
$$
Now as \cite{BG01} if we can find a coordinate transformation $u=\me^Bv$ to conjugate the above system
 into 
$
\mi\dot v=(A^++P^+)v,
$
where $A=\diag(\lambda_j(\omega))_{j\geq 1}$ and $A^+$ is still diagonal and $P^+$ still in $\mathcal{M}_{\alpha,\beta}$, which is higher order than $P$,  one finishes one step of iteration. As \cite{BG01}, one should seek the solutions of 
the homological equation 
$$
[A,B]-\mi\dot B=\widetilde{A}-P+R. 
$$ 
From the standard Fourier analysis the center of the reducibility problem is to estimate the measure of the following set 
\begin{eqnarray}\label{measuresetsetion1}
\bigcup\limits_{|k|\leq K, i>j\geq 1}\{\omega\in \mathcal{D}_0: \big| \langle k, \omega \rangle+\lambda_i(\omega)-\lambda_j(\omega)\big|<\kappa(1+|i-j|)\}. 
\end{eqnarray}
We recall some existent methods to estimate the above set.  In \cite{GT2011}, since the perturbation potential $|V(x, \theta)|\leq \frac{c}{\langle x\rangle^{\delta_1}}$ with $\delta_1>0$, it follows that $\lambda_i=\nu_i+\epsilon_0 i^{-\delta_2}$ with some $\delta_2>0$.  
The measure estimate for (\ref{measuresetsetion1}) is standard in this case, which is close to 1-D wave equations(see \cite{ChYou00}).   In \cite{WL2017}, the estimate is similar since $\lambda_i=\nu_i+\epsilon_0 (\ln i)^{-\delta_3}$ with some positive $\delta_3$ depending on $n$ if we choose small divisor conditions suitably.   \\
\indent Another way to estimate the above set (\ref{measuresetsetion1}) is called quasi-T\"oplitz method. For the 1-D derivative wave equation Berti, Biasco and Procesi \cite{BBP13}(\cite{PrXu13}) proved that 
$$\lambda_j= j+a_{+}(\xi)+\frac{m}{2j}+\mathcal O(\gamma^{2/3}/j)\  {\rm for}\ j\geq \mathcal O(\gamma^{-1/3}).  $$   
When $i>j>\mathcal O(|k|^{\tau} \gamma^{-1/3})$, they could obtain
\[
|\langle  k, \omega(\xi)\rangle+\lambda_i-\lambda_j|=|\langle  k, \omega(\xi)\rangle +i-j+ \frac{m(i-j)}{2ij}+\mathcal O(\gamma^{2/3}/j)|
\]
and the following estimate is close to 1-D  wave equation since they could impose first order Melnikov conditions $|\langle  k, \omega(\xi)\rangle+h|\geq 2\gamma^{2/3}/|k|^{\tau}$ for any $h\in \Z$. \\
\indent In the concerned problem  we will use a different method to estimate the set (\ref{measuresetsetion1}).  The trick is to use the discrete difference matrix and its element decay from the norm estimates.  In fact from the iteration step(see (\ref{EstimateAProposition})) one can prove  
\[
|\Delta(A-A_0)_j^j| = |(A-A_0)_{j+1}^{j+1}-(A-A_0)_{j}^{j}|\leq Cj^{-2\beta},
\]
which is proved  to be important for measure estimate as the following.  
Note 
\begin{align*}
&|k\cdot\omega+\lambda_i-\lambda_j|\ge|k\cdot\omega+\nu_i-\nu_j|-
|(A-A_0)_i^i-(A-A_0)_j^j|\\
\ge&|k\cdot\omega+i-j|-\sum_{l=j}^{i-1}|\Delta(A-A_0)_l^l|
\ge |k\cdot\omega+i-j|-\frac{C|i-j|}{j^{2\beta}}.
\end{align*}
The following estimate is similar as above. \\
\indent  In the end let us review the previous works on the reducibility and the behaviors of solutions in Sobolev spaces. Reducibility for PDEs in high dimension was initiated by Eliasson-Kuksin \cite{EK2009}. We can refer to \cite{GP19}, \cite{LiangW2019} and \cite{LW2021} for higher-dimensional QHO with bounded potential. The  reducibility result for n-D QHO with polynomial perturbations was first set up  in \cite{BGMR2018}.  Montalto \cite{Mon19} obtained the first reducibility result for linear wave equations with unbounded perturbations on $\mathbb T^d$, which can be applied to the linearized Kirchhoff equation in higher dimension.  For transportation equations with unbounded perturbations, see Bambusi, Langella and Montalto \cite{BLM18} (\cite{FGiMP19}). Recently, there are also some research papers \cite{FGr19, FGN19} with  a linear Schr\"odinger equation on Zoll manifold with unbounded potential.  By implementing the above techniques  the KAM-type results of  quasi - linear PDEs such as  incompressible Euler flows in 3-D   and forced Kirchhoff equation on $\mathbb T^d$  have been established in \cite{BaMon21} and \cite{CoMon18} respectively.  \\
\indent The reducibility results usually imply the boundedness of Sobolev norms. We refer to the papers \cite{BGMR2018, Del2014,FaRa2020, GrYa2000,LZZ2020, Mas2018,Th2020} for the growth rate of the solutions of QHO including the upper-lower bound.  There are also many literatures, e.g. \cite{BGMR2019,BLM2021, BM2019, MR2017}, which are closely relative to the upper growth bound of the solution in Sobolev space. \\
 \indent The rest paper will be organized as follows. In Section 2,  a new reducibility theorem is presented. In Section 3, we check all the hypotheses of the  reducibility theorem  are satisfied which follows  the main theorem. In the following section we prove the  reducibility theorem.  Finally, the appendix contains some technical Lemmas. 
 
 \section{Reducibility theorem}
In this section we will state an abstract reducibility theorem for a quasi-periodic in time system of the form $\mi\dot u=\big(A+\epsilon P(\omega t)\big)u$ with $ A=\diag(\nu_i)_{i\ge1}$.
\subsection{Setting}
We first introduce some spaces and work out various properties.\\
\textbf{Infinite matrices}. In the introduction we have introduced the Banach space $\mathcal{M}_{\alpha}$ and its subspace $\mathcal{M}_{\alpha,\beta}$.  As \cite{GT2011} we denote by $\mathcal M_{\alpha+}$ the set of infinite matrices $A:\mathbb N_0\times\mathbb N_0\mapsto\mathbb C$ that satisfy $|A|_{\alpha+}:=\sup_{i,j\in\mathbb N_0}(1+|i-j|)^{\alpha+1}|A_i^j|<\infty$. Similarly,  we denote by $\mathcal M_{\alpha+,\beta}$ the subspace of $\mathcal M_{\alpha+}$: an infinite matrix $A$ belongs to $\mathcal M_{\alpha+,\beta}$ if  
$$|A|_{\alpha+,\beta}: =\sup_{i,j\in\mathbb N_0}(1+|i-j|)^{\alpha+1}|A_i^j|+\sup_{i,j\in\mathbb N_0}(1+|i-j|)(ij)^\beta|\Delta A_i^j|<\infty. $$ We remark that $\mathcal M_{\alpha+}\subset\mathcal M_\alpha$ and $\mathcal M_{\alpha+,\beta}\subset\mathcal M_{\alpha,\beta}$ for any $\alpha,\beta\in\mathbb R$.\par
The following structural lemma is proved in Appendix \ref{ProofAppendix}. 
\begin{Lemma}\label{Algebra}
Let $\alpha\ge\beta>0$. There exists a constant $C\equiv C(\alpha,\beta)>0$ such that\\
(i). Let $A,B\in\mathcal M_{\alpha+,\beta}$. Then $AB$ belongs to $\mathcal M_{\alpha+,\beta}$ and $|AB|_{\alpha+,\beta}\le C|A|_{\alpha+,\beta}|B|_{\alpha+,\beta}$.\\
(ii). Let $A\in\mathcal M_{\alpha,\beta}$ and $B\in\mathcal M_{\alpha+,\beta}$. Then $AB$ and $BA$ belong to $\mathcal M_{\alpha,\beta}$ and
\[
|AB|_{\alpha,\beta},|BA|_{\alpha,\beta}\le C|A|_{\alpha,\beta}|B|_{\alpha+,\beta}.
\]
(iii). Let $A\in\mathcal M_{\alpha+}$. Then we have $\|A\|_{\mathcal{B}(\ell_s^2)}\le C|A|_{\alpha+}$ for any $s\in(-2\alpha-1,2\alpha+1)$.\\
(iv). Let $A\in\mathcal M_\alpha$. Then we have
\begin{alignat*}{5}
&\|A\|_{\mathcal{B}(\ell_1^2,\ell_{-1}^2)}\le C|A|_\alpha,&&\text{if }\alpha\in(0,\tfrac12],\\
&\|A\|_{\mathcal{B}(\ell_0^2,\ell_s^2)}\le C|A|_\alpha,~\forall~s<2\alpha-2,&\quad&\text{if }\alpha\in(\tfrac12,1],\\
&\|A\|_{\mathcal{B}(\ell_0^2)}\le C|A|_\alpha,&&\text{if }\alpha\in(1,\infty).
\end{alignat*}
\end{Lemma}
\noindent
\textbf{Parameter}. In this paper $\omega$ will play the role of a parameter belonging to $\mathcal{D}_0:=[0,2\pi)^n$. All the constructed maps will depend on $\omega$ with $\mathcal C^1$ regularity. When a map is only defined on a Cantor subset of $\mathcal D_0$ the regularity has to be understood in the Whitney sense.\par
Let $\sigma>0$ and $\mathcal D\subset\mathcal{D}_0$. We denote by $\mathcal M_{\alpha,\beta}(\mathcal{D},\sigma)$ the set of $\mathcal C^1$ mappings $\mathcal D\times\mathbb T_\sigma^n\ni(\omega,\theta)\mapsto P(\omega,\theta)\in\mathcal{M}_{\alpha,\beta}$ which is real analytic in $\theta\in\mathbb{T}_\sigma^n$ and is equipped with the norm $\|P\|_{\alpha,\beta}^{\mathcal{D},\sigma}:=\sup\limits_{\substack{|\Im\theta|<\sigma\\ \omega\in\mathcal{D},l=0,1}}|\partial_\omega^lP(\omega,\theta)|_{\alpha,\beta}$. The subset of $\mathcal M_{\alpha,\beta}(\mathcal{D},\sigma)$ formed by the mappings $B$ such that $B(\omega,\theta)\in\mathcal M_{\alpha+,\beta}$ is denoted by $\mathcal M_{\alpha+,\beta}(\mathcal D,\sigma)$ and equipped with the norm 
$\|B\|_{\alpha+,\beta}^{\mathcal{D},\sigma}:=\sup\limits_{\substack{|\Im\theta|<\sigma\\ \omega\in\mathcal{D},l=0,1}}|\partial_\omega^lB(\omega,\theta)|_{\alpha+,\beta}$. The space of mappings $A\in\mathcal M_{\alpha,\beta}(\mathcal D,\sigma)$ that are indenpendent of $\theta$ will be denoted by $\mathcal M_{\alpha,\beta}(\mathcal D)$ and equipped with the norm 
$\|A\|_{\alpha,\beta}^\mathcal D:=\sup\limits_{\omega\in\mathcal D,l=0,1}
|\partial_\omega^lA(\omega)|_{\alpha,\beta}$.
\subsection{Hypothesis on the spectrum}
Now we formulate our hypothesis on $\nu_i,~i\in\mathbb N_0$.
\begin{Hypothesis1}[asymptotics]
Assume that there exist some absolute positive constants $c_0,c_1$ and $\delta$ such that for all $i,j\in\mathbb N_0$
\[
|\nu_i-\nu_j|\ge c_0|i-j|\text{ and }|\nu_{i+1}-\nu_i+\nu_j-\nu_{j+1}|\le\frac{c_1|i-j|}{(ij)^\delta}.
\]
\end{Hypothesis1}
\begin{Hypothesis2}[second Melnikov condition in measure]
There exist some absolute positive constants $\tau_1,\tau_2$ and $c_2$ such that the following holds: for  each $\gamma>0$ and $K\ge1$ there exists a closed subset $\mathcal{D}=\mathcal{D}(\gamma,K)\subset\mathcal{D}_0$ satisfying $\meas(\mathcal{D}_0\setminus\mathcal{D})\le c_2\gamma^{\tau_1}K^{\tau_2}$ such that for all $\omega\in\mathcal D$, all $k\in\mathbb Z^n$ with $0<|k|\le K$ and all $i,j\in\mathbb N_0$ we have 
\[
|k\cdot\omega+\nu_i-\nu_j|\ge\gamma(1+|i-j|). 
\]
\end{Hypothesis2}
\subsection{The reducibility theorem} Consider the non-autonomous system in $\mathcal{H}$:
\begin{equation}\label{OriginalSystem}
\mi\dot u=\big( A+\epsilon P(\omega_1t,\omega_2t,\cdots,\omega_nt)\big)u,
\end{equation}
where 
$A=\diag(\nu_i)_{i\ge1}$ and $P$ is defined for $i,j\in\mathbb N_0$ by 
\begin{equation}\label{PDefinition}
P_j^i(\omega t)=\int_\mathbb R V(x,\omega t)h_i(x)h_j(x)\,\md x,
\end{equation} 
and we define $$\mathcal{H}=
\begin{cases}
	\ell_1^2,&\text{for } 0<\alpha\leq 1/2,\\
	\ell_0^2,&\text{for }\alpha>1/2.
\end{cases}$$
\begin{Theorem}\label{ReducibilityTheorem}
Assume that $(\nu_i)_{i\ge1}$ satisfies Hypotheses H1, H2 and that the perturbation matrix $P$ is hermitian and  $P\in\mathcal M_{\alpha,\beta}(\mathcal D_0,\sigma)$ for some $0<\beta\le\min\{\alpha,\delta\}$ and $\sigma>0$.  There exists $\epsilon_*>0$ such that for all $0\le\epsilon<\epsilon_*$ there is a subset $\mathcal{D}_\epsilon\subset\mathcal{D}_0=[0,2\pi)^n$ of asymptotically full measure such that for all $\omega\in\mathcal{D}_\epsilon$, the system \eqref{OriginalSystem} reduces in $\mathcal{H}$ to an autonomous system 
\begin{equation}\label{ReducedSystem}
	\mi\dot v= A^\infty(\omega)v,\quad A^\infty(\omega)=\diag(\lambda_i^\infty(\omega))_{i\ge1} 
\end{equation}
by a coordinate transformation $u=\Phi_{\omega, \epsilon}(\omega t)v$,  where $\Phi_{\omega, \epsilon}(\theta)$ is unitary on $\ell_0^2$ and  real analytic in  $\theta\in\mathbb T_{\sigma/2}^n$ and  $\lambda_i^\infty(\omega)\in\mathbb R$ is $\mathcal C^1$ smooth in $\omega$ and close to $\nu_i$ for all $i\in\mathbb N_0$.\par
More precisely, let $p\in[0,2\alpha+1)$ then there is a $C>0$ such that
\begin{equation}\label{ReducibilityTheoremEstimate}
\begin{gathered}
\meas(\mathcal{D}_0\setminus\mathcal{D}_\epsilon)\le C\epsilon^{\frac{\iota_1}{17}}\text{ with }\iota_1=\tfrac\beta{\beta+1}\max\{\tau_1,1\},\\
|\lambda_i^\infty-\nu_i|\le C\epsilon,\quad\forall~i\ge1,\\
\|\Phi_{\omega, \epsilon}^{\pm1}(\theta)-\id\|_{\mathcal{B}(\ell_p^2)}\le C\epsilon^\frac23,\quad\forall~\theta\in\mathbb T^n_{\sigma/2}.
\end{gathered}
\end{equation}\par
\end{Theorem}

\section{Application to quantum harmonic oscillator on $\mathbb R$}
In this section we will prove Theorem \ref{MainTheorem} as a direct corollary of the reducibility theorem.  Before the checks on all the hypotheses, we expand  the Schr\"odinger equation \eqref{Schrodinger} on the  Hermite basis $(h_i)_{i\ge1}$ and then it 
is equivalent to the system \eqref{OriginalSystem} with $\nu_i=2i-1$ for $i\ge1$. 
\subsection{Verification of the spectrum}\label{Verification of the spectrum}
For 1-D quantum harmonic oscillator we have $\nu_i=2i-1$ for $i\ge1$. Then H1 is verified with $c_0=c_1=\delta=1$($\delta$ can be any positive number). Besides, we invoke Lemma \ref{AuxiliaryMeasure} to confirm H2 for $\tau_1=1$ and $\tau_2=n+1$. Then we only remain to verify that $P$ defined by \eqref{PDefinition} belongs to $\mathcal{M}_{\alpha,\beta}(\mathcal{D}_0,\sigma)$ for some $\alpha,\beta$ with $0<\beta\le\min\{\alpha,\delta\}$, which follows from the subsequent lemma.\\
\indent Before the proof we introduce  two important operators and some facts.  
From Lemma 3.1 in \cite{Chodosh11},  one obtains  for any $i\geq 1$, 
\begin{align}
\partial_xh_i(x)&=-\sqrt{\tfrac i2}h_{i+1}(x)+\sqrt{\tfrac{i-1}{2}}h_{i-1}(x) \label{equality1},\\
xh_i(x)&=\sqrt{\tfrac i2}h_{i+1}(x)+\sqrt{\tfrac{i-1}{2}}h_{i-1}(x) \label{equality2}.
\end{align} 
Denote $T=\partial_x+x$ and $T^\dagger=-\partial_x+x$.  From (\ref{equality1}),  (\ref{equality2}) and a straightforward computation, we have 
$T^*=T^\dagger,TT^\dagger=H+\id$ and 
\[
Th_i(x)=\sqrt{2(i-1)}h_{i-1}(x),\quad T^\dagger h_i(x)=\sqrt{2i}h_{i+1}(x),\quad TT^\dagger h_i(x)=2ih_i(x).
\]

\begin{Lemma}\label{PVerification}
Assume that a $\mathcal C^1$ smooth function $V(x,\theta)$  analytically in $\theta$ extends to $\mathbb T_\sigma^n$ for $\sigma>0$ and fulfills \eqref{PotentialAssumption},  then the  perturbation matrix $ P(\theta)$ defined by \eqref{PDefinition} is real analytic from $\mathbb{T}_\sigma^n$ to $\mathcal{M}_{1,\frac12}$.
\end{Lemma}
\begin{proof}
To simplify notation we will denote by $\langle\cdot,\cdot\rangle$ the standard scalar product in $L^2(\mathbb R )$ and by superscript ``\,$'$\," the partial derivative\,(w.r.t.$x$)\,. Clearly, $| P_i^i(\theta)|=|\langle Vh_i,h_i\rangle|\le C$ by \eqref{PotentialAssumption}. We hereafter let $i\ne j$. Since $Hh_i=(2i-1)h_i,i\ge1$, from \eqref{PDefinition} we have
\begin{align*}
	&(2i-2j) P_i^j(\theta)=\langle VHh_i,h_j\rangle-\langle Vh_i,Hh_j\rangle=-\langle V\partial_xh_i',h_j\rangle+\langle Vh_i,\partial_xh_j'\rangle\\
	=&\langle V'h_i',h_j\rangle+\langle Vh_i',h_j'\rangle-\langle V'h_i,h_j'\rangle-\langle Vh_i',h_j'\rangle=\langle V'h_i',h_j\rangle-\langle V'h_j',h_i\rangle.
\end{align*}
In addition, $\langle x\rangle^{-1}\partial_x$ is a pseudo-differential operator of order 0, which implies the boundedness from $L^2$ into itself (see \cite{Bam2017,Bam2018,Chodosh11,Shubin1987} for details). Also, by \eqref{PotentialAssumption} we obtain 
\[
|\langle V'h_i',h_j\rangle|\le\sup_{\substack{\|f\|_{L^2}=\|g\|_{L^2}=1}}|\langle V'(x)\langle x\rangle\langle x\rangle^{-1}\partial_x f,g\rangle|
\le\|V'(x)\langle x\rangle\|_{\mathcal B(L^2)}\cdot\|\langle x\rangle^{-1}\partial_x\|_{\mathcal B(L^2)}\le C,
\]
where $\langle x\rangle=\sqrt{1+x^2}$. Similarly, $|\langle V'h_j',h_i\rangle|\le C$ by \eqref{PotentialAssumption}. It follows that there is a $C>0$ such that for all $i,j\ge1$
 \begin{equation}\label{PijEstimate}
 (1+|i-j|)| P_i^j(\theta)|\le C,\quad\forall~\theta\in\mathbb T_\sigma^n.
 \end{equation}
We now commence the estimates on  $\Delta P_i^j(\theta)$. From the definition we have
\begin{align*}
&\Delta P_i^j(\theta)=\langle Vh_{i+1},h_{j+1}\rangle-\langle Vh_i,h_j\rangle
=\frac1{2\sqrt{ij}}\langle VT^\dagger h_i,T^\dagger h_j\rangle-\langle Vh_i,h_j\rangle\\
=&\frac1{2\sqrt{ij}}\langle TVT^\dagger h_i,h_j\rangle-\langle Vh_i,h_j\rangle
=\frac1{2\sqrt{ij}}\langle VTT^\dagger h_i,h_j\rangle-\langle Vh_i,h_j\rangle+\frac1{2\sqrt{ij}}\langle V'T^\dagger h_i,h_j\rangle\\
=&\left(\sqrt{\tfrac ij}-1\right)\langle Vh_i,h_j\rangle+\frac1{2\sqrt{ij}}\langle V'T^\dagger h_i,h_j\rangle
=\frac{i-j}{\sqrt j(\sqrt i+\sqrt j)}\langle Vh_i,h_j\rangle +\frac1{2\sqrt{ij}}\langle V'T^\dagger h_i,h_j\rangle.
\end{align*}
We likewise have $|\langle V'T^\dagger h_i,h_j\rangle|\le C$ by \eqref{PotentialAssumption}. Collecting the last estimate and \eqref{PijEstimate} leads to that for all $i,j\ge1$
\begin{eqnarray}\label{chafenguji}
\sqrt{ij}|\Delta P_i^j(\theta)|\le(1+|i-j|)|\langle Vh_i,h_j\rangle|+|\langle V'T^\dagger h_i,h_j\rangle|\le C,\quad\forall~\theta\in\mathbb T_\sigma^n.
\end{eqnarray}
Combining \eqref{PijEstimate} with \eqref{chafenguji}, 
it follows that $ P(\theta)\in\mathcal M_{1,\frac12}$. The rest proof is clear. \qedhere
\end{proof}

\subsection{Proof of Theorem \ref{MainTheorem}}
Expanded on Hermite basis $(h_i)_{i\ge1}$ the Schr\"odinger equation \eqref{Schrodinger} reads as the non-autonomous system \eqref{OriginalSystem} with $\nu_i=2i-1$. By subsection \ref{Verification of the spectrum}  if set $\alpha=\delta=1$ and $\beta=\frac12$,  then all the hypotheses  of  Theorem \ref{ReducibilityTheorem} are satisfied.  Since $\alpha=1$, 
the system \eqref{OriginalSystem} reduces in $\ell_0^2$ to an autonomous system (\ref{ReducedSystem}). 
More precisely,  in the new coordinates given by Theorem \ref{ReducibilityTheorem} $u=\Phi_{\omega,\epsilon}(\omega t)v$, the original system \eqref{OriginalSystem} becomes an autonomous system $\mi\dot v= A^\infty v$ with $ A^\infty=\diag(\lambda_i^\infty)_{i\ge1}$. We solve the Cauchy problem \eqref{Schrodinger} subject to the initial datum $\psi(x,0)=\sum_{i\ge1}u_i(0)h_i(x)\in L^2$, which reads 
\[
\psi(x,t)=\sum_{i\ge1}u_i(t)h_i(x)\in L^2\text{ with }u(t)=\Phi_{\omega,\epsilon}(\omega t)\me^{-\mi t A^\infty}\Phi_{\omega,\epsilon}^{-1}(0)u(0).
\]
Now define the coordinate transformation $\Psi_{\omega,\epsilon}(\theta)\in\mathcal{B}(L^2)$ by 
\[
\Psi_{\omega,\epsilon}(\theta)\left(\sum_{i\ge1}v_ih_i(x)\right)=\sum_{i\ge1}\big(\Phi_{\omega,\epsilon}(\theta)v\big)_ih_i(x),
\]
and let $p\in[0,3)$.  If the initial data $\zeta(0)=\sum\limits_{j\geq 1} v_j(0)h_j(x) \in \mathcal{H}^p$, then one has $\zeta(\cdot,t)=\sum\limits_{j\geq 1}v_j(t)h_j(x)\in \mathcal C^0(\mathbb R,\mathcal H^p)$ satisfies $\mi\partial_t\zeta=H^\infty\zeta$ with $H^\infty=\diag(\lambda_i^\infty)_{i\ge1}$, where $v_j(t)=e^{-{\rm i} \lambda_j^{\infty}t} v_j(0)$ for all $j\geq 1$.  For $\omega\in \mathcal{D}_{\epsilon}$ if  denote  $\psi(\cdot,t)=\Psi_{\omega,\epsilon}(\omega t)\zeta(\cdot,t)$ then it is clear that $\psi(\cdot, t)\in\mathcal C^0(\mathbb R,\mathcal{H}^p)$ and satisfies \eqref{Schrodinger}. The other side is similar. Corollary \ref{SolutionEstimate} and Corollary \ref{PurePoint} are similar as \cite{GT2011}. \qed

\section{Proof of reducibility theorem}
In this section we will prove Theorem \ref{ReducibilityTheorem} by KAM methods with the general strategy in the opening subsection.
\subsection{General strategy} Consider the non-autonomous system of the form 
\begin{equation}\label{FirstSystem}
\mi\dot u=( A+P)u,
\end{equation}
where $ A$ is diagonal. Note that at the beginning of the KAM iteration $ A= A_0=\diag(\nu_i)_{i\ge1}$ is independent of $\omega$, and the perturbation matrix $P(\theta)$ is of size $\mathcal O(\epsilon)$. We search for a coordinate transformation $u=\me^Bv$ with $B$ of size $\mathcal O(\epsilon)$ to conjugate the system \eqref{FirstSystem} into 
\begin{equation*}
\mi\dot v=( A^++P^+)v,
\end{equation*}
where $ A^+$ is still diagonal and $\epsilon$-close to $ A$, and the new perturbation $P^+$ is of size $\mathcal O(\epsilon^2)$ at least formally. More precisely, as a consequence of the construction and 
Bambusi \cite{Bam2018} or \cite{BG01} we have that
\begin{align*}
& A^++P^+=& A+[ A,B]-\mi\dot B+P+\int_0^1\me^{-sB}[(1-s)([ A,B]-\mi\dot B+P)+sP,B]\me^{sB}\,\md s.
\end{align*}
So to achieve the goal at least formally we should solve the homological equation 
$$[ A,B]-\mi\dot B+P= A^+- A+R$$
 with $R$ of size $\mathcal{O}(\epsilon^2)$. Concretely, we hereafter let $ A^+= A+\widetilde  A$ with $\widetilde{ A}=\diag(\bar{P})=\mathcal O(\epsilon)$\,(The bar denotes Fourier average) and  
\begin{align*}
P^+=&R+\int_0^1\me^{-sB}[(1-s)(\widetilde{ A}+R)+sP,B]\me^{sB}\,\md s=\mathcal O(\epsilon^2).
\end{align*}
Repeating iteratively the same procedure with $ A^+$ instead of $ A$, we will construct a coordinate transformation $u=\Phi_{\omega,\epsilon} v$ to conjugate the original non-autonomous system \eqref{OriginalSystem} into an autonomous system $\mi\dot v= A^\infty v$ with $ A^\infty$ diagonal and $\epsilon$-close to $ A_0$.
\subsection{Homological equation}
Following the previous strategy, in this section we will consider a homological equation of the form $[ \mathcal A,B]-\mi\dot B+P=$ remainder with $\mathcal A$ diagonal and close to $A_0=\diag(\nu_i)_{i\ge1}$ and $P\in\mathcal M_{\alpha,\beta}$ with $0<\beta\le\min\{\alpha,\delta\}$. Then we will construct a solution $B\in\mathcal M_{\alpha+,\beta}$ by the following proposition.
\begin{Proposition}\label{HomoEqnSolution}
Denote $\mathcal D\subset\mathcal D_0$. Let $\mathcal D\ni\omega\mapsto  \mathcal A(\omega)=\diag(\lambda_i(\omega))_{i\ge1}$ be a $\mathcal C^1$ mapping that verifies for some $0<c\le\min\{\frac{c_0}4,\frac14\}$
\begin{equation}\label{ConditionProposition}
\| \mathcal A- A_0\|_{\alpha,\beta}^\mathcal D\le c.
\end{equation}
Let $P\in\mathcal M_{\alpha,\beta}(\mathcal D,\sigma)$ with $0<\beta\le\min\{\alpha,\delta\}$ be hermitian, and $0<\kappa\le\gamma\le\min\{\frac{c_0}4,1\}$ and $K\ge1$. There exists a subset $\mathcal D'=\mathcal D'(\kappa,K)\subset\mathcal D$ fulfilling 
\begin{equation}\label{MeasureProposition}
\meas(\mathcal D\setminus\mathcal D')\le C\kappa^{\iota_1}K^{\iota_2}
\end{equation}
and $\mathcal C^1$ mappings $\widetilde{ \mathcal A}:\mathcal D'\mapsto\mathcal M_{\alpha,\beta}$ diagonal, $R:\mathcal D'\times\mathbb T^n_{\sigma'}\mapsto\mathcal M_{\alpha,\beta}$ hermitian, $B:\mathcal D'\times\mathbb T^n_{\sigma'}\mapsto\mathcal M_{\alpha+,\beta}$ anti-hermitian, and analytic in $\theta$, such that
\begin{equation}\label{HomoEqnProposition}
[ \mathcal A,B]-\mi\dot B=\widetilde{ \mathcal A}-P+R
\end{equation}
and for any $0<\sigma'<\sigma$
\begin{gather}
\|\widetilde{ \mathcal A}\|_{\alpha,\beta}^{\mathcal{D}'}\le\|P\|_{\alpha,\beta}^{\mathcal{D},\sigma},\label{EstimateAProposition}\\
\|R\|_{\alpha,\beta}^{\mathcal{D}',\sigma'}\le\frac{C\me^{-\frac{K}2(\sigma-\sigma')}}{(\sigma-\sigma')^n}\|P\|_{\alpha,\beta}^{\mathcal D,\sigma},\label{EstimateRProposition}\\
\|B\|_{\alpha+,\beta}^{\mathcal{D}',\sigma'}\le\frac{CK}{\kappa^4(\sigma-\sigma')^n}\|P\|_{\alpha,\beta}^{\mathcal D,\sigma}.\label{EstimateBProposition}
\end{gather}
Here the constant $C>0$ depends on $\alpha,\beta,\delta,c,c_0,c_1,c_2$ and $\iota_1=\frac\beta{\beta+1}\max\{\tau_1,1\},\iota_2=\max\{\tau_2,n+1\}$ with all the explicit parameters shown in Hypotheses H1, H2.
\end{Proposition}
\begin{proof}
Expanding in Fourier series, the homological equation \eqref{HomoEqnProposition} reads
\begin{equation}\label{HomoEqnFourier}
L\widehat{B}(k)=\delta_{k,0}\widetilde{ \mathcal A}-\widehat{P}(k)+\widehat{R}(k),
\end{equation}
where $\delta_{k,l}$ denotes the Kronecker symbol and $L:=L(\omega,k)$ is the linear operator, acting on $\mathcal M_0$, defined by
\[
LM=k\cdot\omega M+[ \mathcal A(\omega),M].
\]
Taking its matrix elements between the eigenvalues of $ \mathcal A$ the equation \eqref{HomoEqnFourier} becomes
\begin{equation}\label{HomoEqnElement}
k\cdot\omega\widehat{B}_i^j(k)+(\lambda_i-\lambda_j)\widehat{B}_i^j(k)=\delta_{k,0}\widetilde{ \mathcal A}_i^j-\widehat{P}_i^j(k)+\widehat{R}_i^j(k).
\end{equation}
First solve this equation when $|k|+|i-j|=0$\,(i.e. $k=0$ and $i=j$)\,by defining
\[
\widehat{B}_i^i(0)=0,~\widehat{R}_i^i(0)=0\text{ and }\widetilde{ \mathcal A}_i^i=\widehat{P}_i^i(0).
\]
Then we set $\widetilde{ \mathcal A}_i^j=0$ for $i\ne j$ in such a way $\widetilde{ \mathcal A}\in\mathcal M_{\alpha,\beta}$ and fulfills $|\widetilde{ \mathcal A}|_{\alpha,\beta}\le|\widehat{P}(0)|_{\alpha,\beta}$. The estimates of the derivatives\,(w.r.t.\,$\omega$)\,are obtained by differentiating the expression of $\widetilde{ \mathcal A}$. Taking all the estimations on $\widetilde{ \mathcal A}$ leads to \eqref{EstimateAProposition}.\par
Now let us consider the remainder case when $|k|+|i-j|>0$. We solve equation \eqref{HomoEqnElement} by defining for $i,j\ge1$
\begin{gather}
\widehat{R}_i^j(k)=
\begin{cases}
0,&\text{for }|k|\le K,\\
\widehat{P}_i^j(k),&\text{for }|k|>K;
\end{cases}\label{DefinitionR}\\
\widehat{B}_i^j(k)=
\begin{cases}
0,&\text{for }|k|>K\text{ or }|k|+|i-j|=0,\\
\frac{-\widehat{P}_i^j(k)}{k\cdot\omega+\lambda_i-\lambda_j},&\text{in the other cases}.
\end{cases}\label{DefinitionB}
\end{gather}
Before the estimations on such matrices, first note that with this resolution, $R$ and $\widetilde{ \mathcal A}$ are hermitian but $B$ is anti-hermitian because $ \mathcal A$ and $P$ are hermitian. From equation \eqref{DefinitionR}, a canonical Fourier analysis leads to that for any $|\Im\theta|<\sigma'$
\begin{equation}\label{EstimateRtheta}
|R(\theta)|_{\alpha,\beta}\le\frac{C\me^{-\frac K2(\sigma-\sigma')}}{(\sigma-\sigma')^n}\sup_{|\Im\theta|<\sigma}|P(\theta)|_{\alpha,\beta}.
\end{equation}
In order to estimate $B$, we will use Lemma \ref{KeyLemma} stated in Appendix \ref{LemmaAppendix}. Let us face the small divisors 
$
k\cdot\omega+\lambda_i-\lambda_j,~i,j\in\mathbb N_0.
$
We will distinguish two cases depending on whether $k=0$ or not.\par
(a). The case $k=0$. In this case $i\ne j$ and from \eqref{ConditionProposition} and Hypothesis H1 we have
\[
|k\cdot\omega+\lambda_i-\lambda_j|\ge|\nu_i-\nu_j|-|\lambda_i-\nu_i|-|\lambda_j-\nu_j|\ge c_0|i-j|-2c\ge\kappa(1+|i-j|).
\]
The last estimate and \eqref{ConditionProposition} allow us to use Lemma \ref{KeyLemma} to conclude that for $i\ne j$
\[
(1+|i-j|)^{\alpha+1}|\widehat{B}_i^j(0)|+(1+|i-j|)(ij)^\beta|\Delta\widehat{B}_i^j(0)|
\le\frac{C}{\kappa^2}|\widehat{P}(0)|_{\alpha,\beta}.
\]
Together with $\widehat{B}_i^i(0)=0$ we obtain that $\widehat{B}(0)\in\mathcal M_{\alpha+,\beta}$ and 
\begin{equation}\label{EstimateBzero}
|\widehat{B}(0)|_{\alpha+,\beta}\le\frac{C}{\kappa^2}|\widehat{P}(0)|_{\alpha,\beta}.
\end{equation}\par
(b). The case $k\ne0$. Concretely, in this case we only solve the main terms of Fourier series truncated by $K$ and throw out a small set of parameter $\omega$. Using Hypothesis H2, for each $\gamma>0$ there exists a subset $\mathcal D_1:=\mathcal D(2\gamma,K)$ fulfilling $\meas(\mathcal D_0\setminus\mathcal D_1)\le C\gamma^{\tau_1}K^{\tau_2}$ such that for all $\omega\in\mathcal D_1$
\[
|k\cdot\omega+\nu_i-\nu_j|\ge2\gamma(1+|i-j|),
\quad\forall~i,j\in\mathbb N_0,~\forall~k\in\mathbb Z^n\text{ with }0<|k|\le K.
\]
Without loss of generality, let $i\ge j$ in the following discussion. By \eqref{ConditionProposition} the last equation implies that 
\begin{align*}
&|k\cdot\omega+\lambda_i-\lambda_j|\ge|k\cdot\omega+\nu_i-\nu_j|-
|( \mathcal A- A_0)_i^i-( \mathcal A- A_0)_j^j|\\
=&|k\cdot\omega+\nu_i-\nu_j|-
|\sum_{l=j}^{i-1}\Delta( \mathcal A- A_0)_l^l|
\ge|k\cdot\omega+\nu_i-\nu_j|-\sum_{l=j}^{i-1}|\Delta( \mathcal A- A_0)_l^l|\\\ge&2\gamma(1+|i-j|)-\sum_{l=j}^{i-1}\frac{c}{l^{2\beta}}\ge2\gamma(1+|i-j|)-\frac{c|i-j|}{j^{2\beta}}\ge\left(2\gamma-\frac{c}{j^{2\beta}}\right)(1+|i-j|)\\
\ge&\gamma(1+|i-j|)\ge\kappa(1+|i-j|),\quad\underline{\text{provided }i\ge j\ge(\tfrac{c}{\gamma})^{\frac1{2\beta}}}.
\end{align*}
Now let $j\le(\frac{c}{\gamma})^{\frac1{2\beta}}$. In fact, the inequality $|i-j|\ge CK$ turns out that $|k\cdot\omega+\lambda_i-\lambda_j|\ge\kappa(1+|i-j|)$. Then we only need to face the small divisors
\[
k\cdot\omega+\lambda_i-\lambda_j,\quad\text{ for }0\le i-j\le CK,~j\le(\tfrac{c}{\gamma})^{\frac1{2\beta}}\text{ and }k\in\mathbb Z^n\text{ with }0<|k|\le K.
\]
Clearly, one has $i\le CK\gamma^{-\frac1{2\beta}}$. Since $|\partial_\omega(k\cdot\omega)\cdot\tfrac{k}{|k|}|=|k|\ge1$, then by \eqref{ConditionProposition} one has 
\[
|\partial_\omega(k\cdot\omega+\lambda_i(\omega)-\lambda_j(\omega))\cdot\tfrac{k}{|k|}|\ge\tfrac{|k|}2,
\]
which allows us to use Lemma \ref{AuxiliaryMeasure} to conclude that for any 
$i,j\in\mathbb N_0$ and $k\in\mathbb Z^n$ with $0<|k|\le K$ one has
\[
|k\cdot\omega+\lambda_i-\lambda_j|\ge\kappa(1+|i-j|),
\]
expect a subset $F_{i,j,k}\subset\mathcal D$ fulfilling $\meas\left(F_{i,j,k}\right)\le\frac{C\kappa(1+|i-j|)}{|k|}$. Denoting $\mathcal D_2$ be the union of $F_{i,j,k}$ for $i\le CK\gamma^{-\frac1{2\beta}},~j\le C\gamma^{-\frac1{2\beta}}$ with $|i-j|\le CK$ and $0<|k|\le K$ we have
\[
\meas(\mathcal D_2)\le C\kappa\gamma^{-\frac1\beta}K^{n+1}.
\]
Let $\mathcal D'=\mathcal D_1\bigcup\left(\mathcal D\setminus\mathcal D_2\right)$ and $\kappa=\gamma^{1+\frac1\beta},\iota_1=\frac{\beta}{\beta+1}\max\{\tau_1,1\},\iota_2=\max\{\tau_2,n+1\}$, then
\[
\meas(\mathcal D\setminus\mathcal D')\le\meas(\mathcal D_0\setminus\mathcal D_1)+\meas(\mathcal D_2)\le C\gamma^{\tau_1}K^{\tau_2}+C\gamma K^{n+1}\le C\kappa^{\iota_1}K^{\iota_2}.
\]
By construction for all $\omega\in\mathcal D'$ we have $|k\cdot\omega+\lambda_i-\lambda_j|\ge\kappa(1+|i-j|)$, which allows us to use Lemma \ref{KeyLemma} to conclude that for all $0<|k|\le K$
\[
\widehat{B}(k)\in\mathcal M_{\alpha+,\beta}\text{ and }|\widehat{B}(k)|_{\alpha+,\beta}\le\frac{C}{\kappa^2}|\widehat{P}(k)|_{\alpha,\beta}.
\]
Collecting the last estimate and \eqref{EstimateBzero} we have for any $|\Im\theta|<\sigma'$
\begin{equation}\label{EstimateBtheta}
|B(\theta)|_{\alpha+,\beta}\le\frac{C}{\kappa^2(\sigma-\sigma')^n}\sup_{|\Im\theta|<\sigma}|P(\theta)|_{\alpha,\beta}.
\end{equation}
For the proceeding estimates on derivatives\,(w.r.t.$\omega$)\,we differentiate  \eqref{HomoEqnFourier} to get
\[
L\partial_\omega\widehat{B}(\omega,k)=\delta_{k,0}\partial_\omega\widetilde{ \mathcal A}(\omega)-\left(\big(\partial_\omega L\big)\widehat{B}(\omega,k)+\partial_\omega\widehat{P}(\omega,k)\right)+\partial_\omega\widehat{R}(\omega,k),
\]
which is an equation of the same type as \eqref{HomoEqnFourier} for $\partial_\omega\widehat{B}(\omega,k)$ and $\partial_\omega\widehat{R}(\omega,k)$ with $\widehat{P}(k)$ replaced by $Q(\omega,k):=\big(\partial_\omega L\big)\widehat{B}(\omega,k)+\partial_\omega\widehat{P}(\omega,k)$. 
We can solve this equation by defining
\begin{gather*}
\partial_\omega\widehat{B}(\omega,k)=\chi_{|k|\le K}(k)L^{-1}(\omega,k)\left(\delta_{k,0}\partial_\omega\widetilde{ \mathcal A}(\omega)-Q(\omega,k)\right)\\
\partial_\omega\widehat{R}(\omega,k)=\chi_{|k|>K}(k)Q(\omega,k)=\chi_{|k|>K}(k)\partial_\omega\widehat{P}(\omega,k).
\end{gather*}
By \eqref{ConditionProposition} we get for all 
$|k|\le K,~|(\partial_\omega L)\widehat{B}(\omega,k)|_{\alpha,\beta}\le CK|\widehat{B}(\omega,k)|_{\alpha+,\beta}$, which follows
\begin{align*}
&|\partial_\omega\widehat{B}(\omega,k)|_{\alpha+,\beta}\le\|L^{-1}(\omega,k)\|\cdot\left(|(\partial_\omega L)\widehat{B}(\omega,k)|_{\alpha,\beta}+|\partial_\omega\widehat{P}(\omega,k)|_{\alpha,\beta}\right)\\
\le&\frac{C}{\kappa^2}\left(\frac K{\kappa^2}|\widehat{P}(\omega,k)|_{\alpha,\beta}+|\partial_\omega\widehat{P}(\omega,k)|_{\alpha,\beta}\right)
\le\frac{CK}{\kappa^4}\left(|\widehat{P}(\omega,k)|_{\alpha,\beta}+|\partial_\omega\widehat{P}(\omega,k)|_{\alpha,\beta}\right).
\end{align*}
A canonical Fourier analysis leads to that for any $|\Im\theta|<\sigma'$
\begin{gather*}
	|\partial_\omega B(\omega,\theta)|_{\alpha+,\beta}\le\frac{CK}{\kappa^4(\sigma-\sigma')^n}\left(\sup_{|\Im\theta|<\sigma}|P(\omega,\theta)|_{\alpha,\beta}+\sup_{|\Im\theta|<\sigma}|\partial_\omega P(\omega,\theta)|_{\alpha,\beta}\right),\\
	|\partial_\omega R(\omega,\theta)|_{\alpha,\beta}\le\frac{C\me^{-\frac K2(\sigma-\sigma')}}{(\sigma-\sigma')^n}\sup_{|\Im\theta|<\sigma}|\partial_\omega P(\omega,\theta)|_{\alpha,\beta}.
\end{gather*}
Collecting the last two estimates and equations \eqref{EstimateRtheta}, \eqref{EstimateBtheta} concludes \eqref{EstimateRProposition}, \eqref{EstimateBProposition}.\qedhere
\end{proof}

\subsection{The KAM iteration}
In this section we will present our iterative KAM procedures following the previous strategy in the opening subsection. Let us begin with the initial non-autonomous system 
\begin{equation}\label{InitialSystemKAM}
\mi\dot u=\big( A_0+P_0(\omega t)\big)u,
\end{equation}
where $ A_0=A=\diag(\nu_i)_{i\ge1}$ fulfilling Hypotheses H1 and H2, $\omega\in\mathcal D_0$ and $P_0=\epsilon P\in\mathcal M_{\alpha,\beta}(\mathcal D_0,\sigma_0)$ with $0<\beta\le\min\{\alpha,\delta\}$ and $\sigma_0=\sigma$. We will build iteratively coordinate transformation $u=\phi_{m+1}v$ to conjugate the system $\mi\dot u=( A_m+P_m)u$ with $ A_m=\diag\big(\lambda_i^{(m)}\big)_{i\ge1}$ and $P_m\in\mathcal M_{\alpha,\beta}(\mathcal D_m,\sigma_m)$ into $\mi\dot v=( A_{m+1}+P_{m+1})v$ as follows: assume that the construction has been built up to step $m\ge0$ then\par
(i). we use Proposition \ref{HomoEqnSolution} to construct $B_{m+1}(\omega,\theta)$ solution of the homological equation
\begin{equation}\label{SchemeFirst}
	[ A_m,B_{m+1}]-\mi\dot B_{m+1}=\widetilde{ A}_m-P_m+R_m,\quad(\omega,\theta)\in\mathcal D_{m+1}\times\mathbb T^n_{\sigma_{m+1}},
\end{equation}
with $\widetilde{ A}_m(\omega),R_m(\omega,\theta)$ defined for $(\omega,\theta)\in\mathcal D_{m+1}\times\mathbb T^n_{\sigma_{m+1}}$ by
\begin{gather}
	\widetilde{ A}_m(\omega)=\left(\delta_{ij}\left(\widehat{P}_m(\omega,0)\right)_i^j\right)_{i,j\ge1},\\
	R_m(\omega,\theta)=\sum_{|k|>K_{m+1}}\widehat{P}_m(\omega,k)\me^{\mi k\cdot\theta};
\end{gather}\indent
(ii). we define $ A_{m+1},P_{m+1}$ for $(\omega,\theta)\in\mathcal D_{m+1}\times\mathbb T^n_{\sigma_{m+1}}$ by
\begin{gather}
 A_{m+1}= A_m+\widetilde{ A}_m,\\
P_{m+1}=R_m+\int_0^1\me^{-sB_{m+1}}[(1-s)(\widetilde{ A}_m+R_m)+sP_m,B_{m+1}]\me^{sB_{m+1}}\,\md s.
\label{SchemeFinal}
\end{gather}
Notice that the coordinate transformation $\phi_{m+1}:v\mapsto u=\me^{B_{m+1}}v$ and by construction if $ A_m$ and $P_m$ are hermitian, so are $\widetilde{ A}_m,R_m$ and $ A_{m+1}$, by resolution of homological equation $B_{m+1}$ is anti-hermitian which implies $P_{m+1}$ is hermitian. Following the general scheme \eqref{SchemeFirst}--\eqref{SchemeFinal} we obtain that the  coordinate transformation $\Phi_m=\phi_1\circ\phi_2\circ\cdots\circ\phi_m:v\mapsto u=\me^{B_1}\circ\me^{B_2}\circ\cdots\circ\me^{B_m}v$ conjugates the original system \eqref{InitialSystemKAM} into $\mi\dot v=( A_m+P_m)v$ with $ A_m$ diagonal and $P_m\in\mathcal M_{\alpha,\beta}(\mathcal D_m,\sigma_m)$. At step $m$ the Fourier series are truncated at order $K_m$ and the small divisors are controlled by $\kappa_m$. Now we specify the choice of all the parameters. First let $\sigma_0=\sigma$ and $\|P_0\|_{\alpha,\beta}^{\mathcal{D}_0,\sigma_0}\le\epsilon_0$. Then choose for $m\ge1$
\[
\epsilon_m=\epsilon_{m-1}^{4/3},~\kappa_m=\epsilon_{m-1}^{1/16},~\sigma_{m-1}-\sigma_m=\frac{\sigma_0}{2C_*}m^{-2},~K_m=2(\sigma_{m-1}-\sigma_m)^{-1}\ln\epsilon_{m-1}^{-1},
\]
where $C_*=\sum_{m\ge1}m^{-2}$.
\begin{Lemma}\label{KAMiterationLemma}
Let $0<\beta\le\min\{\alpha,\delta\}$ and $\iota_1=\frac{\beta}{\beta+1}\max\{\tau_1,1\}$. There is $\epsilon_*\ll1$ depending on $\sigma,n,\alpha,\beta,\tau_1,\tau_2$ and $ A_0$ such that for all $0\le\epsilon<\epsilon_*$ the following holds for all $m\ge1$: there exist $\mathcal D_m\subset\mathcal D_{m-1},B_m\in\mathcal M_{\alpha+,\beta}(\mathcal D_m,\sigma_m)$ and $P_m\in\mathcal M_{\alpha,\beta}(\mathcal D_m,\sigma_m)$ such that\par
(i). for all $p\in(-2\alpha-1,2\alpha+1)$ the coordinate transformation 
\[
\phi_m(\omega,\theta):\ell_p^2\ni v\mapsto\me^{B_m}v=u\in\ell_p^2,\quad\forall~(\omega,\theta)\in\mathcal D_m\times\mathbb T^n_{\sigma_m}
\]
is linear (unitary on $\ell_0^2$) isomorphism conjugating the system at the $(m-1)^\mth$  KAM step $\mi\dot u=( A_{m-1}+P_{m-1})u$ into the system at the $m^\mth$ step $\mi\dot v=( A_m+P_m)v$;\par
(ii). we have 
\begin{gather}
\meas(\mathcal D_{m-1}\setminus\mathcal D_m)\le\epsilon_{m-1}^{\frac{\iota_1}{17}},\label{MeasureEstimateIteration}\\
\|\widetilde{ A}_{m-1}\|_{\alpha,\beta}^{\mathcal D_m}\le\epsilon_{m-1},
\label{AEstimateIteration}\\
\|B_m\|_{\alpha+,\beta}^{\mathcal D_m,\sigma_m}\le\epsilon_{m-1}^{2/3},
\label{BEstimateIteration}\\
\|P_m\|_{\alpha,\beta}^{\mathcal D_m,\sigma_m}\le\epsilon_m,\label{PEstimateIteration}
\end{gather}
and for all $p\in(-2\alpha-1,2\alpha+1)$ the coordinate transformation satisfies
\begin{equation}\label{TransformationEstimate}
\|\phi_m(\omega,\theta)-\id\|_{\mathcal B(\ell_p^2)}\le\epsilon_{m-1}^{2/3},\quad \forall~(\omega,\theta)\in\mathcal D_m\times\mathbb T_{\sigma_m}^n.
\end{equation}
\end{Lemma}
\begin{proof}
We proceed by induction applying Proposition \ref{HomoEqnSolution}. At the first step the initial $A=A_0$ which implies \eqref{ConditionProposition}, then we use Proposition \ref{HomoEqnSolution} to construct $B_1,\widetilde{ A}_0,R_0$ and $\mathcal D_1,\sigma_1$ such that for all $(\omega,\theta)\in\mathcal D_1\times\mathbb T_{\sigma_1}^n$
\[
[ A_0,B_1]-\mi\dot B_1=\widetilde{ A}_0-P_0+R_0.
\] 
Due to \eqref{MeasureProposition} we have 
\[
\meas(\mathcal D_0\setminus\mathcal D_1)\le C\kappa_1^{\iota_1}K_1^{\iota_2}
\le C\epsilon_0^{\frac{\iota_1}{16}}\left((\sigma_0-\sigma_1)^{-1}\ln\epsilon_0^{-1}\right)^{\iota_2}\le\epsilon_0^{\frac{\iota_1}{17}}.
\]
In view of \eqref{EstimateBProposition} we have 
\[
\|B_1\|_{\alpha+,\beta}^{\mathcal D_1,\sigma_1}\le\frac{CK_1}{\kappa_1^4(\sigma_0-\sigma_1)^n}\|P_0\|_{\alpha,\beta}^{\mathcal D_0,\sigma_0}
\le\frac{C\ln\epsilon_0^{-1}}{\epsilon_0^{1/4}(\sigma_0-\sigma_1)^{n+1}}\epsilon_0
\le C\sigma_0^{-(n+1)}\epsilon_0^{3/4}\ln\epsilon_0^{-1}\le\epsilon_0^{2/3}.
\]
Thus, from the definition and Lemma \ref{Algebra} we have for all $p\in(-2\alpha-1,2\alpha+1)$
\[
\|\phi_1(\omega,\theta)-\id\|_{\mathcal B(\ell_p^2)}=\|\me^{B_1}-\id\|_{\mathcal B(\ell_p^2)}\le C\me^{C\|B_1\|_{\alpha+,\beta}^{\mathcal D_1,\sigma_1}}\|B_1\|_{\alpha+,\beta}^{\mathcal D_1,\sigma_1}\le\epsilon_0^{2/3},\quad\forall~(\omega,\theta)\in\mathcal D_1\times\mathbb T_{\sigma_1}^n.
\]
Collecting equations \eqref{EstimateAProposition} and \eqref{EstimateRProposition} leads to that $\|\widetilde{ A}_0\|_{\alpha,\beta}^{\mathcal D_1}\le\epsilon_0$ and 
\begin{equation}\label{R0}
\|R_0\|_{\alpha,\beta}^{\mathcal D_1,\sigma_1}\le\frac{C\me^{-\frac{K_1}{2}(\sigma_0-\sigma_1)}}{(\sigma_0-\sigma_1)^n}\|P_0\|_{\alpha,\beta}^{\mathcal D_0,\sigma_0}\le C\sigma_0^{-n}\epsilon_0^2\le\frac12\epsilon_0^{4/3}.
\end{equation}
Besides, by Lemma \ref{Algebra} we have for $s\in[0,1]$
\[
\|[(1-s)(\widetilde{ A}_0+R_0)+sP_0,B_1]\|_{\alpha,\beta}^{\mathcal D_1,\sigma_1}\le C\|P_0\|_{\alpha,\beta}^{\mathcal D_0,\sigma_0}\cdot\|B_1\|_{\alpha+,\beta}^{\mathcal D_1,\sigma_1}\le C\epsilon_0^{5/3}.
\]
It follows that
\[
\left\|\int_0^1\me^{-sB_1}[(1-s)(\widetilde{ A}_0+R_0)+sP_0,B_1]\me^{sB_1}\,\md s\right\|_{\alpha,\beta}^{\mathcal D_1,\sigma_1}\le\frac12\epsilon_0^{4/3}.
\]
Collecting the last estimate and equation \eqref{R0} concludes $\|P_1\|_{\alpha,\beta}^{\mathcal D_1,\sigma_1}\le\epsilon_0^{4/3}=\epsilon_1$.\\
Now assume that we have verified Lemma \ref{KAMiterationLemma} up to step $m>1$, then we go from step $m$ to step $m+1$. Clearly, from the assumption we have
\[
\| A_m- A_0\|_{\alpha,\beta}^{\mathcal D_m}=\|\sum_{l=0}^{m-1}\widetilde{ A}_l\|_{\alpha,\beta}^{\mathcal D_m}\le\sum_{l=0}^{m-1}\|\widetilde  A_l\|_{\alpha,\beta}^{\mathcal D_{l+1}}\le\sum_{l=0}^{m-1}\epsilon_l\le2\epsilon_0,
\]
which allows us to apply Proposition \ref{HomoEqnSolution} to construct $B_{m+1},\widetilde{ A}_m,R_m$ and $\mathcal D_{m+1},\sigma_{m+1}$ such that for all $(\omega,\theta)\in\mathcal D_{m+1}\times\mathbb T^n_{\sigma_{m+1}}$
\[
[ A_m,B_{m+1}]-\mi\dot B_{m+1}=\widetilde{ A}_m-P_m+R_m.
\]
Due to \eqref{MeasureProposition} we have 
\[
\meas(\mathcal D_m\setminus\mathcal D_{m+1})\le C\kappa_{m+1}^{\iota_1}K_{m+1}^{\iota_2}\le C\epsilon_m^{\frac{\iota_1}{16}}\left((\sigma_m-\sigma_{m+1})^{-1}\ln\epsilon_m^{-1}\right)^{\iota_2}
\le\epsilon_m^{\frac{\iota_1}{17}}.
\]
In view of \eqref{EstimateBProposition} we have
\[
\|B_{m+1}\|_{\alpha+,\beta}^{\mathcal D_{m+1},\sigma_{m+1}}\le\frac{CK_{m+1}}{\kappa_{m+1}^4(\sigma_m-\sigma_{m+1})^n}\|P_m\|_{\alpha,\beta}^{\mathcal D_m,\sigma_m}\le\frac{C(\ln\epsilon_m^{-1})\epsilon_m^{3/4}}{(\sigma_m-\sigma_{m+1})^{n+1}}\le\epsilon_m^{2/3}.
\]
Thus, from the definition and Lemma \ref{Algebra} we have for all $p\in(-2\alpha-1,2\alpha+1)$
\[
\|\phi_{m+1}(\omega,\theta)-\id\|_{\mathcal B(\ell_p^2)}\le C\me^{C\|B_{m+1}\|_{\alpha+,\beta}^{\mathcal D_{m+1},\sigma_{m+1}}}\|B_{m+1}\|_{\alpha+,\beta}^{\mathcal D_{m+1},\sigma_{m+1}}\le\epsilon_m^{2/3}.
\]
Collecting equations \eqref{EstimateAProposition} and \eqref{EstimateRProposition} leads to that $\|\widetilde{ A}_m\|_{\alpha,\beta}^{\mathcal D_{m+1}}\le\epsilon_m$ and
\begin{equation}\label{Rm}
	\|R_m\|_{\alpha,\beta}^{\mathcal D_{m+1},\sigma_{m+1}}\le\frac{C\me^{-\frac{K_{m+1}}{2}(\sigma_m-\sigma_{m+1})}}{(\sigma_m-\sigma_{m+1})^n}\|P_m\|_{\alpha,\beta}^{\mathcal D_{m},\sigma_m}\le C\sigma_0^{-n}(m+1)^{2n}\epsilon_m^2\le\frac12\epsilon_m^{4/3}.
\end{equation}
Besides, by Lemma \ref{Algebra} we have for $s\in[0,1]$
\[
\|[(1-s)(\widetilde{ A}_m+R_m)+sP_m,B_{m+1}]\|_{\alpha,\beta}^{\mathcal D_{m+1},\sigma_{m+1}}\le C\|P_m\|_{\alpha,\beta}^{\mathcal D_m,\sigma_m}\cdot\|B_{m+1}\|_{\alpha+,\beta}^{\mathcal D_{m+1},\sigma_{m+1}}\le C\epsilon_m^{5/3}.
\]
It follows that
\[
\left\|\int_0^1\me^{-sB_{m+1}}[(1-s)(\widetilde{ A}_m+R_m)+sP_m,B_{m+1}]\me^{sB_{m+1}}\,\md s\right\|_{\alpha,\beta}^{\mathcal D_{m+1},\sigma_{m+1}}\le\frac12\epsilon_m^{4/3}.
\]
Collecting the last estimate and equation  \eqref{Rm} leads to  $\|P_{m+1}\|_{\alpha,\beta}^{\mathcal D_{m+1},\sigma_{m+1}}\le\epsilon_m^{4/3}=\epsilon_{m+1}$. By induction, we complete the proof. \qedhere
\end{proof}

\subsection{Proof of Theorem \ref{ReducibilityTheorem}}
Let $\mathcal D_\epsilon=\cap_{m\ge0}\mathcal D_m$. Owing \eqref{MeasureEstimateIteration} it is a Borel set satisfying
\begin{equation}\label{MeasureEstimateProof}
\meas(\mathcal D_0\setminus\mathcal D_\epsilon)\le\sum_{m\ge0}\epsilon_{m}^{\frac{\iota_1}{17}}\le2\epsilon_0^{\frac{\iota_1}{17}}.
\end{equation}
Clearly, $\sigma_\infty=\sigma_0-\sum_{m\ge1}(\sigma_{m-1}-\sigma_m)=\frac{\sigma}2$. Then we let $(\omega,\theta)\in{\mathcal D}_\epsilon\times\mathbb T^n_{\sigma/2}$ and $p\in(-2\alpha-1,2\alpha+1)$ in the following discussion. Also, $P_m\to0$ as $m\to\infty$ by \eqref{PEstimateIteration}. Moreover, writing $ A^\infty:=\diag(\lambda_i^\infty)_{i\ge1}$ with $\lambda_i^\infty=\lim\limits_{m\to\infty}\lambda_i^{(m)}$ for $i\ge1$, by \eqref{AEstimateIteration} one has 
\[
\| A^\infty- A_0\|_{\alpha,\beta}^{\mathcal D_\epsilon}\le\sum_{m\ge0}\epsilon_m\le2\epsilon_0,
\]
which follows that
\begin{equation}\label{DifferenceAProof}
|\lambda_i^\infty-\nu_i|\le2\epsilon_0,\quad\forall~i\ge1.
\end{equation}
At last, we only remain to estimate the coordinate transformation $\Phi_{\omega,\epsilon}(\theta)$. First we present two auxiliary lemmas below.
\begin{Lemma}\label{TansformationLemma}
Denote $\Phi_m=\phi_1\circ\phi_2\circ\cdots\circ\phi_m$ for $m\ge1$, then we have
\[
\|\Phi_m(\omega,\theta)-\id\|_{\mathcal B(\ell_p^2)}\le\sum_{l=0}^{m-1}2\epsilon_l^{2/3}.
\]
\end{Lemma}
\begin{proof}
We proceed by induction applying equation \eqref{TransformationEstimate} in Lemma \ref{KAMiterationLemma}. Clearly, \eqref{TransformationEstimate} implies that $\|\Phi_1(\omega,\theta)-\id\|_{\mathcal B(\ell_p^2)}\le\epsilon_0^{2/3}\le2\epsilon_0^{2/3}$. Now we assume that Lemma \ref{TansformationLemma} has been verified up to step $m>1$, then we go from step $m$ to $m+1$. Since
\[
\Phi_{m+1}-\id=\Phi_m\circ\phi_{m+1}-\id=\Phi_m\circ(\phi_{m+1}-\id)+\Phi_m-\id,
\]
then by above assumption and equation \eqref{TransformationEstimate} we have
\[
\|\Phi_{m+1}-\id\|_{\mathcal B(\ell_p^2)}\le\|\Phi_{m}\|_{\mathcal B(\ell_p^2)}\cdot\|\phi_{m+1}-\id\|_{\mathcal B(\ell_p^2)}+\|\Phi_m-\id\|_{\mathcal B(\ell_p^2)}\le
\sum_{l=0}^{m}2\epsilon_l^{2/3}.
\]
By induction, we complete the proof.\qedhere
\end{proof}
\begin{Lemma}
The coordinate transformation $\left(\Phi_m(\omega,\theta)\right)_{m\ge1}$ is a Cauchy sequence in $\mathcal B(\ell_p^2)$. Letting $\Phi_{\omega,\epsilon}(\theta)\in\mathcal B(\ell_p^2)$ be the limit mapping we have 
\begin{equation}\label{TransformationEstimateProof}
\|\Phi_{\omega,\epsilon}(\theta)-\id\|_{\mathcal B(\ell_p^2)}\le 4\epsilon_0^{2/3},
\end{equation}
which follows that $\Phi_{\omega,\epsilon}(\theta)$ is analytic in $\theta\in\mathbb T^n_{\sigma/2}$.
\end{Lemma}
\begin{proof}
Since $\Phi_{m+1}-\Phi_m=\Phi_m\circ\phi_{m+1}-\Phi_m=\Phi_m\circ(\phi_{m+1}-\id)$, then by Lemma \ref{TansformationLemma} and equation \eqref{TransformationEstimate} we have
\[
\|\Phi_{m+1}-\Phi_m\|_{\mathcal B(\ell_p^2)}\le\|\Phi_m\|_{\mathcal B(\ell_p^2)}\cdot\|\phi_{m+1}-\id\|_{\mathcal B(\ell_p^2)}\le2\epsilon_m^{2/3}.
\]
Given $m_2>m_1\ge1$, we have
\[
\|\Phi_{m_2}-\Phi_{m_1}\|_{\mathcal B(\ell_p^2)}\le\sum_{l=m_1}^{m_2-1}\|\Phi_{l+1}-\Phi_l\|_{\mathcal B(\ell_p^2)}\le\sum_{l=m_1}^{m_2-1}2\epsilon_l^{2/3}\le4\epsilon_{m_1}^{2/3}\to0\text{ as }m_1\to\infty.
\]
It follows the Cauchy sequence. Then we invoke Lemma \ref{TansformationLemma} to complete the proof.\qedhere
\end{proof}
By now we have obtained all the estimates \eqref{ReducibilityTheoremEstimate} of the reducibility theorem, which follows from equations \eqref{MeasureEstimateProof}, \eqref{DifferenceAProof} and \eqref{TransformationEstimateProof}. \\
\indent In the end we will prove the system \eqref{OriginalSystem} reduces in $\ell_1^2$ or $\ell_0^2$  to an autonomous system \eqref{ReducedSystem}.  
By Lemma \ref{Algebra} (iii),(iv) we have $\mathcal M_{\alpha}\subset\mathcal B(\ell_1^2,\ell_{-1}^2)$ and $\mathcal M_{\alpha+}\subset\mathcal B(\ell_s^2),~\forall~s\in[-1,1]$, provided $0<\alpha\le1/2$.  Besides, the coordinate transformation $u=\Phi_mv$ conjugates the original system $\mi\dot u=( A+\epsilon P)u,u\in\ell_1^2$ into $\mi\dot v=( A_m+P_m)v,v\in\ell_1^2$. These follow the identity in $\mathcal B(\ell_1^2,\ell_{-1}^2)$ below
\[
 A_m+P_m=\Phi_m^{-1}( A+\epsilon P)\Phi_m-\mi\Phi_m^{-1}\partial_t\Phi_m.
\]
By construction   $\Phi_{\omega,\epsilon}=\lim_{m\to\infty}\me^{B_1}\circ\me^{B_2}\circ\cdots\circ\me^{B_m}$. Letting $m\to\infty$, one obtains the reducibility identity in $\mathcal B(\ell_1^2,\ell_{-1}^2)$
\begin{equation}\label{ReducibilityIdentity}
 A^\infty(\omega)=\Phi_{\omega,\epsilon}^{-1}\big( A+\epsilon P(\omega t)\big)\Phi_{\omega,\epsilon}-\mi\Phi_{\omega,\epsilon}^{-1}\partial_t\Phi_{\omega,\epsilon},\quad\omega\in\mathcal D_\epsilon.
\end{equation}
Consequently,  if $v(t)\in\mathcal C^0(\mathbb R,\ell_1^2)\bigcap\mathcal C^1(\mathbb R,\ell_{-1}^2)$ satisfies \eqref{ReducedSystem} and define $u=\Phi_{\omega,\epsilon}v$, then by the reducibility identity \eqref{ReducibilityIdentity} one gets
\[
\mi\dot u=\Phi_{\omega,\epsilon}\mi\dot v+\mi\left(\partial_t\Phi_{\omega,\epsilon}\right)v
=\left(\Phi_{\omega,\epsilon} A^\infty\Phi_{\omega,\epsilon}^{-1}+\mi\left(\partial_t\Phi_{\omega,\epsilon}\right)\Phi_{\omega,\epsilon}^{-1}\right)u=\big(A+\epsilon P(\omega t)\big)u.
\]
In addition, from $v(t)\in\mathcal C^0(\mathbb R,\ell_1^2)\bigcap\mathcal C^1(\mathbb R,\ell_{-1}^2)$ we can draw $u(t)\in\mathcal C^0(\mathbb R,\ell_1^2)\bigcap\mathcal C^1(\mathbb R,\ell_{-1}^2)$. Conversely, if $u(t)\in\mathcal C^0(\mathbb R,\ell_1^2)\bigcap\mathcal C^1(\mathbb R,\ell_{-1}^2)$ satisfies \eqref{OriginalSystem} and define $v=\Phi_{\omega,\epsilon}^{-1}u$, then one has $v(t)\in\mathcal C^0(\mathbb R,\ell_1^2)\bigcap\mathcal C^1(\mathbb R,\ell_{-1}^2)$ fulfills \eqref{ReducedSystem} by the identity \eqref{ReducibilityIdentity}. This explains the reducibility in $\ell_1^2$ when $0<\alpha\leq 1/2$. 
\par
On the other hand, if $\alpha>1/2$,  by Lemma \ref{Algebra} (iii),(iv) we have $\mathcal M_\alpha\subset\mathcal B(\ell_0^2,\ell_{-2}^2)$ and $\mathcal M_{\alpha+}\subset\mathcal B(\ell_s^2),~\forall~s\in[-2,2]$.  
It implies the system \eqref{OriginalSystem} reduces in $\ell_0^2$ in the sense that if $v(t)\in\mathcal C^0(\mathbb R,\ell_0^2)\bigcap\mathcal C^1(\mathbb R,\ell_{-2}^2)$ fulfills \eqref{ReducedSystem} if and only if  $u=\Phi_{\omega,\epsilon}v\in\mathcal C^0(\mathbb R,\ell_0^2)\bigcap\mathcal C^1(\mathbb R,\ell_{-2}^2)$ fulfills \eqref{OriginalSystem} because the above identity \eqref{ReducibilityIdentity} holds in $\mathcal B(\ell_0^2,\ell_{-2}^2)$.  This completes the proof of Theorem \ref{ReducibilityTheorem}. \qed

\appendix
\section{Proof of Lemma \ref{Algebra}}\label{ProofAppendix}
Recall that $\alpha\ge\beta>0$, which will be used occasionally.\par 
(i). Since $A,B\in\mathcal M_{\alpha+,\beta}$ then by $\sum_{l\ge1}\le\left(\sum_{1+|i-l|\ge\frac{1+|i-j|}2}+\sum_{1+|j-l|\ge\frac{1+|i-j|}2}\right)$ one has
\[
(1+|i-j|)^{\alpha+1}|(AB)_i^j|\le C|A|_{\alpha+}|B|_{\alpha+}.
\]
In addition, from the definition we have
\begin{equation}\label{DeltaABij}
\Delta(AB)_i^j=\sum_{l\ge1}A_{i+1}^lB_l^{j+1}-\sum_{l\ge1}A_i^lB_l^j
=A_{i+1}^1B_1^{j+1}+\sum_{l\ge1}\Delta A_i^l\cdot B_{l+1}^{j+1}+\sum_{l\ge1}A_i^l\cdot\Delta B_l^j.
\end{equation}
We estimate it term by term. First see that
\[
|A_{i+1}^1B_1^{j+1}|\le\frac{|A|_{\alpha+}|B|_{\alpha+}}{(1+i)^{\alpha+1}(1+j)^{\alpha+1}}\le\frac{|A|_{\alpha+}|B|_{\alpha+}}{ij(ij)^{\alpha}}\le\frac{|A|_{\alpha+}|B|_{\alpha+}}{(1+|i-j|)(ij)^\beta}.
\]
Then turn to the remainder terms. We have that
\[
|\sum_{l\ge1}\Delta A_i^l\cdot B_{l+1}^{j+1}|,~
|\sum_{l\ge1}A_i^l\cdot\Delta B_l^j|\le\frac{C|A|_{\alpha+,\beta}|B|_{\alpha+,\beta}}{(1+|i-j|)(ij)^\beta},
\]
which follow from
\begin{equation}\label{iSeriesJbeta}
\sum_{l\ge1}\frac1{(1+|j-l|)^{\alpha+1}l^\beta}\le\frac C{j^\beta},\quad\sum_{l\ge1}\frac1{(1+|i-l|)(1+|j-l|)^\alpha l^\beta}\le\frac C{j^\beta}.
\end{equation}
Therefore, we obtain $|AB|_{\alpha+,\beta}\le C|A|_{\alpha+,\beta}|B|_{\alpha+,\beta}$. 
\par
(ii). Recall that $A\in\mathcal M_{\alpha,\beta},B\in\mathcal M_{\alpha+,\beta}$. Clearly, we have
\[
(1+|i-j|)^\alpha|(AB)_i^j|\le C|A|_\alpha|B|_{\alpha+}.
\]
We remain to estimate $\Delta(AB)_i^j$. Recall the equation \eqref{DeltaABij} and
first observe that 
\[
|A_{i+1}^1B_1^{j+1}|\le\frac{|A|_\alpha|B|_{\alpha+}}{(1+i)^\alpha(1+j)^{\alpha+1}}\le\frac{|A|_\alpha|B|_{\alpha+}}{(ij)^\beta}.
\]
Also, by the equation \eqref{iSeriesJbeta} one has
\[
|\sum_{l\ge1}\Delta A_i^l\cdot B_{l+1}^{j+1}|,~
|\sum_{l\ge1}A_i^l\cdot\Delta B_l^j|\le\frac{C|A|_{\alpha,\beta}|B|_{\alpha+,\beta}}{(ij)^\beta}.
\]
Hence, we have $|AB|_{\alpha,\beta}\le C|A|_{\alpha,\beta}|B|_{\alpha+,\beta}$.
Clearly, repeating the same procedures we likewise obtain that
\(
|BA|_{\alpha,\beta}\le C|A|_{\alpha,\beta}|B|_{\alpha+,\beta}.
\)
\par 
(iii). First let $s\in(0,2\alpha+1)$ and claim that there is a $C>0$ such that
\begin{equation}\label{iiiClaim}
\sum_{j\ge1}\left(\frac ij\right)^{\frac sp}\frac1{(1+|i-j|)^{\alpha+1}}\le C,\quad
\sum_{i\ge1}\left(\frac ij\right)^{\frac sq}\frac1{(1+|i-j|)^{\alpha+1}}\le C,
\end{equation}
for some $p,q>0$ with $\tfrac1p+\tfrac1q=1$ and 
\begin{equation}\label{iiiConjugatePair}
\frac{s-\alpha}{s}<\frac1p\le\frac{\alpha+1}{s}.
\end{equation}
Taking into account $A\in\mathcal M_{\alpha+}$, we calculate for $\xi\in\ell_s^2$ that
\begin{align*}
	&\|A\xi\|_s^2\le\sum_{i\ge1}i^s\left(\sum_{j\ge1}|A_i^j||\xi_j|\right)^2\le|A|_{\alpha+}^2\sum_{i\ge1}\left(\sum_{j\ge1}\frac{(\tfrac ij)^{\frac s2(\frac1p+\frac1q)}j^{\frac s2}|\xi_j|}{(1+|i-j|)^{\alpha+1}}\right)^2\\
	\le&|A|_{\alpha+}^2\sum_{i\ge1}\left(\sum_{j\ge1}\left(\frac ij\right)^{\frac sp}\frac1{(1+|i-j|)^{\alpha+1}}\right)\left(\sum_{j\ge1}\left(\frac ij\right)^{\frac sq}
	\frac{j^s|\xi_j|^2}{(1+|i-j|)^{\alpha+1}}\right)\\
	\le&C|A|_{\alpha+}^2\sum_{i\ge1}\sum_{j\ge1}\left(\frac ij\right)^{\frac sq}\frac{j^s|\xi_j|^2}{(1+|i-j|)^{\alpha+1}}
	=C|A|_{\alpha+}^2\sum_{j\ge1}j^s|\xi_j|^2\sum_{i\ge1}\frac{(\frac ij)^{\frac sq}}{(1+|i-j|)^{\alpha+1}}\\
	\le&C^2|A|_{\alpha+}^2\sum_{j\ge1}j^s|\xi_j|^2=C^2|A|_{\alpha+}^2\|\xi\|_s^2.
\end{align*}
It follows that $\|A\|_{\mathcal B(\ell_s^2)}\le C|A|_{\alpha+}$ for any $s\in(0,2\alpha+1)$. Clearly, \eqref{iiiClaim} holds true when $s=0$, which leads to that $\|A\|_{\mathcal B(\ell_0^2)}\le C|A|_{\alpha+}$. In addition, by \eqref{iiiClaim} we have 
\[
\sum_{i\ge1}\left(\frac ij\right)^{-\frac sp}\frac1{(1+|i-j|)^{\alpha+1}}\le C,~
\sum_{j\ge1}\left(\frac ij\right)^{-\frac sq}\frac1{(1+|i-j|)^{\alpha+1}}\le C,\quad\forall~s\in(0,2\alpha+1),
\]
provided \eqref{iiiConjugatePair} holds, which follows that $\|A\|_{\mathcal B(\ell_{-s}^2)}\le C|A|_{\alpha+}$ for any $s\in(0,2\alpha+1)$.\par
(iv). Recall that $A\in\mathcal M_\alpha$. We will distinguish three cases depending on the size of $\alpha$.\par
(a). The case $\alpha\in(0,\tfrac12]$. We compute for  $\xi\in\ell_1^2$ that
\begin{align*}
&\|A\xi\|_{-1}^2\le\sum_{i\ge1}i^{-1}\left(\sum_{j\ge1}|A_i^j||\xi_j|\right)^2\le|A|_\alpha^2\sum_{i\ge1}i^{-1}\left(\sum_{j\ge1}\frac{j^\frac12|\xi_j|}{j^{\frac12}(1+|i-j|)^\alpha}\right)^2\\
\le&|A|_\alpha^2\sum_{i\ge1}i^{-1}\left(\sum_{j\ge1}\frac1{j(1+|i-j|)^\alpha}\right)\left(\sum_{j\ge1}\frac{j|\xi_j|^2}{(1+|i-j|)^\alpha}\right)\le C|A|_\alpha^2\sum_{i\ge1}\sum_{j\ge1}\frac{j|\xi_j|^2}{i(1+|i-j|)^\alpha}\\
=&C|A|_\alpha^2\sum_{j\ge1}j|\xi_j|^2\sum_{i\ge1}\frac1{i(1+|i-j|)^\alpha}\le C^2|A|_\alpha^2\|\xi\|_1^2.
\end{align*}\par 
(b). The case $\alpha\in(\tfrac12,1]$. Let $s<2\alpha-2$ and there exists a $C>0$ such that
\[
\sum_{j\ge1}\frac1{(1+|i-j|)^{\frac{2\alpha}{p}}}\le C,\quad\sum_{i\ge1}\frac1{i^{-s}(1+|i-j|)^{\frac{2\alpha}{q}}}\le C,
\]
for some $p,q>0$ with $\tfrac1p+\tfrac1q=1$ and $\tfrac{1}{2\alpha}<\tfrac1p<\tfrac{2\alpha-s-1}{2\alpha}$. 
We calculate for any $\xi\in\ell_0^2$ that
\begin{align*}
&\|A\xi\|_s^2\le\sum_{i\ge1}i^s\left(\sum_{j\ge1}|A_i^j||\xi_j|\right)^2\le|A|_\alpha^2\sum_{i\ge1}i^s\left(\sum_{j\ge1}\frac{|\xi_j|}{(1+|i-j|)^{\frac\alpha p+\frac\alpha q}}\right)^2\\
\le&|A|_\alpha^2\sum_{i\ge1}i^s\left(\sum_{j\ge1}\frac1{(1+|i-j|)^{\frac{2\alpha}{p}}}\right)\left(\sum_{j\ge1}\frac{|\xi_j|^2}{(1+|i-j|)^{\frac{2\alpha}{q}}}\right)
\le C|A|_\alpha^2\sum_{i\ge1}\sum_{j\ge1}\frac{i^s|\xi_j|^2}{(1+|i-j|)^{\frac{2\alpha}{q}}}\\
=&C|A|_\alpha^2\sum_{j\ge1}|\xi_j|^2\sum_{i\ge1}\frac1{i^{-s}(1+|i-j|)^{\frac{2\alpha}{q}}}\le C^2|A|_\alpha^2\|\xi\|_0^2,\qquad\underline{~\text{by}~\tfrac{2\alpha}{q}-s>1~}.
\end{align*}\par 
(c). The case $\alpha\in(1,\infty)$. We compute for $\xi\in\ell_0^2$ that
\begin{align*}
&\|A\xi\|_0^2\le\sum_{i\ge1}\left(\sum_{j\ge1}|A_i^j||\xi_j|\right)^2\le|A|_\alpha^2\sum_{i\ge1}\left(\sum_{j\ge1}\frac{|\xi_j|}{(1+|i-j|)^\alpha}\right)^2\\
\le&|A|_\alpha^2\sum_{i\ge1}\left(\sum_{j\ge1}\frac1{(1+|i-j|)^\alpha}\right)\left(\sum_{j\ge1}\frac{|\xi_j|^2}{(1+|i-j|)^\alpha}\right)\le C|A|_\alpha^2\sum_{i\ge1}\sum_{j\ge1}\frac{|\xi_j|^2}{(1+|i-j|)^{\alpha}}\\
=&C|A|_\alpha^2\sum_{j\ge1}|\xi_j|^2\sum_{i\ge1}\frac1{(1+|i-j|)^\alpha}\le C^2|A|_\alpha^2\|\xi\|_0^2.
\end{align*}
Consequently, collecting all the estimates of above three cases leads to the desired results.

\section{Some technical lemmas}\label{LemmaAppendix}
\begin{Lemma}\label{KeyLemma}
	Let $Q\in\mathcal M_{\alpha,\beta}$ and 
	$(\mu_i)_{i\ge1}$ is a sequence of real numbers satisfying 
	\begin{equation}\label{DifferenceLambdaMu}
		|\mu_{i+1}-\mu_i+\nu_i-\nu_{i+1}|\le\frac{c_\mu}{i^{2\beta}}
	\end{equation}
	for a given $c_\mu>0$, where  $0<\beta\le\min\{\alpha,\delta\}$, $\delta$ and $\nu_i$ are shown in Hypothesis H1. 
	If $B(k) =\left(B_i^j(k) \right)_{i,j\ge1}:=\left(\frac{Q_i^j}{k\cdot\omega+\mu_i-\mu_j}\right)_{i,j\ge1}$ defined for $k\in\mathbb Z^n$ with $|k|\le K$, $\omega\in\mathbb R^n$ and the sequence $(\mu_i)_{i\ge1}$ satisfies  
	\begin{equation}\label{SmallDivisorMu}
		|k\cdot\omega+\mu_i-\mu_j|\ge\kappa(1+|i-j|),\quad\forall~i,j\ge1 \text{ and } |k|\le K, 
	\end{equation}
	then for $|k|\le K$, $B(k)\in \mathcal M_{\alpha+,\beta}$. More clearly,  there is a constant $C\equiv2^{\beta+1}(c_\mu+c_1+1)$  such that for $i,j\geq 1$ and $|k|\le K$, 
	\[
	(1+|i-j|)^{\alpha+1}|B_i^j(k)|+(1+|i-j|)(ij)^\beta|\Delta B_i^j(k)|\le\frac{C|Q|_{\alpha,\beta}}{\kappa^2}.
	\]
\end{Lemma}
\begin{proof}
For simplicity we will omit the variable $k$ if there is no confusion. Since $Q\in\mathcal M_{\alpha,\beta}$ with $0<\beta\le\min\{\alpha,\delta\}$,  then from the definition and \eqref{SmallDivisorMu} we have 
\[
	(1+|i-j|)^{\alpha+1}|B_i^j|\le\frac{(1+|i-j|)^{\alpha+1}|Q_i^j|}{|k\cdot\omega+\mu_i-\mu_j|}\le\frac{(1+|i-j|)^\alpha|Q_i^j|}\kappa\le\frac{|Q|_\alpha}\kappa.
\]
Now we commence the estimations on $\Delta B_i^j$. First observe that
\[
	\Delta B_i^j=\frac{\Delta Q_i^j}{k\cdot\omega+\mu_{i+1}-\mu_{j+1}}+\frac{(\mu_i-\mu_j+\mu_{j+1}-\mu_{i+1})Q_i^j}{(k\cdot\omega+\mu_{i+1}-\mu_{j+1})(k\cdot\omega+\mu_i-\mu_j)}\\
	:=\Delta_1+\Delta_2.
\]
Clearly, by \eqref{SmallDivisorMu} we have 
\[
	(1+|i-j|)(ij)^\beta|\Delta_1|\le\frac{(1+|i-j|)(ij)^\beta|\Delta Q_i^j|}{|k\cdot\omega+\mu_{i+1}-\mu_{j+1}|}\le\frac{\sup_{i,j\ge1}(ij)^\beta|\Delta Q_i^j|}\kappa.
\]
Then turn to $\Delta_2$. We hereafter assume $i\ge j$ without loss of genelarity and will distinguish two cases depending upon whether $i\le 2j$ or not.\\ \indent
(a). The case $j\le i\le 2j$. We calculate that
\[
\Delta_2=\frac{\big((\mu_i-\mu_{i+1}+\nu_{i+1}-\nu_i)+(\mu_{j+1}-\mu_j+\nu_j-\nu_{j+1})+(\nu_{j+1}-\nu_j+\nu_i-\nu_{i+1})\big)Q_i^j}{(k\cdot\omega+\mu_{i+1}-\mu_{j+1})(k\cdot\omega+\mu_i-\mu_j)}.
\]
Consequently, by equations \eqref{DifferenceLambdaMu},\eqref{SmallDivisorMu} and Hypothesis H1 we have
\[
	|\Delta_2|\le\frac{\left(\tfrac{c_\mu}{i^{2\beta}}+\tfrac{c_\mu}{j^{2\beta}}+\tfrac{c_1|i-j|}{(ij)^\delta}\right)|Q_i^j|}{\kappa^2(1+|i-j|)^2}\le\frac{\left(c_\mu+2^\beta c_\mu+c_1|i-j|\right)|Q_i^j|}{\kappa^2(1+|i-j|)^2(ij)^\beta}\le
	\frac{2^{\beta+1}(c_\mu+c_1)|Q_i^j|}{\kappa^2(1+|i-j|)(ij)^\beta},
\]
which follows that $(1+|i-j|)(ij)^\beta|\Delta_2|\le\frac{2^{\beta+1}(c_\mu+c_1)}{\kappa^2}|Q|_\alpha$ for $j\le i\le 2j$.\\ \indent
(b). The case $i>2j$. In this case $1+|i-j|=1+i-j\ge\frac i2$. We likewise obtain that
\begin{align*}
	|\Delta_2|&\le\frac{\left(\tfrac{c_\mu}{i^{2\beta}}+\tfrac{c_\mu}{j^{2\beta}}+\tfrac{c_1|i-j|}{(ij)^\delta}\right)|Q_i^j|}{\kappa^2(1+|i-j|)^2}\le\frac{\left(2c_\mu+c_1|i-j|\right)|Q_i^j|}{\kappa^2(1+|i-j|)^2j^\beta}\le
	\frac{2(c_\mu+c_1)|Q_i^j|}{\kappa^2(1+|i-j|)j^\beta},
\end{align*}
which follows that
$(1+|i-j|)(ij)^\beta|\Delta_2|
	\le\frac{2(c_\mu+c_1)i^\beta|Q_i^j|}{\kappa^2}\le\frac{2^{\beta+1}(c_\mu+c_1)}{\kappa^2}|Q|_\alpha$ for $i>2j$.
Collecting all the above estimates concludes that
\[
(1+|i-j|)(ij)^\beta|\Delta_2|\le\frac{2^{\beta+1}(c_\mu+c_1)}{\kappa^2}|Q|_\alpha,\quad\forall~i,j\ge1.
\]
As a consequence, we obtain the desired results.
\end{proof}
\begin{Lemma}\label{AuxiliaryMeasure}
Let $f:[0,1]\mapsto\mathbb R$ be a $\mathcal C^1$ map satisfying $|f'(x)|\ge\varsigma>0$ for all $x\in[0,1]$ then for each $\kappa>0$ we have $\meas\left(\{x\in[0,1]:|f(x)|\le\kappa\}\right)\le\frac{2\kappa}{\varsigma}$.
\end{Lemma}


\begin{thebibliography}{99}

\bibitem{BaMon21}
Baldi, P., Montalto, R.: Quasi - periodic incompressible Euler flows in 3D. Adv. Math. \textbf{384}, 107730 (2021)

\bibitem{Bam2017}
Bambusi, D.:  Reducibility of 1-d Schr\"odinger equation with time quasiperiodic unbounded perturbations. II. Commun. Math. Phys. \textbf{353}, 353-378 (2017)

\bibitem{Bam2018}
Bambusi, D.:  Reducibility of 1-d Schr\"odinger equation with time quasiperiodic unbounded perturbations. I. Trans. Amer. Math. Soc. \textbf{370}, 1823-1865 (2018)

 \bibitem{BG01}
Bambusi, D., Graffi, S.: Time quasi-periodic unbounded perturbations of Schr\"odinger operators and KAM methods. Commun. Math. Phys. \textbf{219}, 465-480 (2001)

\bibitem{BGMR2018}
Bambusi, D., Gr\'{e}bert, B., Maspero, A., Robert, D.:  Reducibility of the quantum harmonic oscillator in $d$-dimensions with polynomial time-dependent perturbation. Anal. \& PDE \textbf{11}, 775-799 (2018)

\bibitem{BGMR2019}
Bambusi, D., Gr\'ebert, B., Maspero, A., Robert, D.: Growth of Sobolev norms for abstract linear Schr\"odinger equations.
J. Eur. Math. Soc. \textbf{23}, 557-583 (2021)

\bibitem{BLM18}
Bambusi, D., Langella, D., Montalto, R.: Reducibility of non-resonant transport equation on $\mathbb T^d$ with unbounded perturbations. Ann. Henri Poincar\'e.  {\bf 20}, 1893-1929 (2019)

\bibitem{BLM2021}
Bambusi, D., Langella, D., Montalto, R.: Growth of Sobolev norms for unbounded perturbations of the Laplacian on flat tori. arXiv:2012.02654 (2020)

\bibitem{BBP13}
Berti, M., Biasco, L., Procesi M.: KAM theory for the Hamiltonian derivative wave equation. Ann. Scient. \'Ec. Norm. Sup. \textbf{46}, 301-373 (2013)

\bibitem{BM2019}
Berti, M., Maspero, A.: Long time dynamics of Schr\"odinger and wave equations on flat tori.  J. Differ. Equ. \textbf{267}, 1167-1200 (2019). 

 \bibitem{ChYou00}
 Chierchia L., You J.: KAM tori for 1D nonlinear wave equations with periodic boundary conditions. Commun. Math. Phys. \textbf{211}, 497-525 (2000) 
 
 
\bibitem{Chodosh11} Chodosh, O.: Infinite matrix representations of isotropic pseudodifferential operators. Methods Appl. Anal. \textbf{18}, 351-372 (2011)
  
\bibitem{Com87}
Combescure, M.: The quantum stability problem for time-periodic perturbations of the harmonic oscillator. Ann. Inst. H. Poincar\'e Phys. Th\'eor. \textbf{47}(1), 63-83 (1987); Erratum: Ann. Inst. H. Poincar\'e Phys. Th\'eor. \textbf{47}(4), 451-454 (1987)

 \bibitem{CoMon18} 
Corsi, L., Montalto, R.: Quasi-periodic solutions for the forced Kirchhoff equation on $\mathbb{T}^d$. Nonlinearity \textbf{31}, 5075-5109 (2018)

\bibitem{Del2014}
Delort, J.-M.: Growth of Sobolev norms for solutions of time dependent Schr\"odinger operators with harmonic oscillator potential. Commun. PDE \textbf{39}, 1-33 (2014)

\bibitem{DLSV2002}
Duclos, P., Lev, O., \v{S}\v{\mbox{t}}ov\'{i}\v{c}ek, P., Vittot, M.: Weakly regular Floquet Hamiltonians with pure point spectrum. Rev. Math. Phys. \textbf{14}, 531-568 (2002).

\bibitem{Eli2011}
Eliasson,  L. H.:  Reducibility for linear quasi-periodic differential equations. Winter School, St Etienne de
Tin\'ee (2011)

\bibitem{EK2009} Eliasson, L. H., Kuksin, S. B.: On reducibility of Schr\"odinger equations with quasiperiodic in time potentials. Commun.  Math. Phys. \textbf{286}, 125-135 (2009) 

\bibitem{EV83}
Enss, V., Veseli\' c, K.: Bound states and propagating states for time-dependent hamiltonians. Ann. Inst. H. Poincar\'e Phys. Th\'eor. \textbf{39}, 159-191 (1983).

\bibitem{FaRa2020} Faou, E., Rapha\"el, P.: On weakly turbulent solutions to the perturbed linear harmonic oscillator.  arXiv: 2006.08206 (2020)

\bibitem{FGiMP19}
Feola, R., Giuliani, F., Montalto, R., Procesi, M.: Reducibility of first order linear operators on tori via Moser's theorem. J. Funct. Anal. \textbf{276}, 932-970 (2019)

\bibitem{FGr19}
Feola, R., Gr\'ebert, B.: Reducibility of Schr\"odinger equation on the sphere. Int. Math. Res. Not. \textbf{0}, 1-39 (2020)

\bibitem{FGN19}
Feola, R., Gr\'ebert, B., Nguyen, T.: Reducibility of Schr\"odinger equation on a Zoll manifold with unbounded potential. J. Math. Phys. \textbf{61}, 071501 (2020) 

\bibitem{GrYa2000}
Graffi, S., Yajima, K.: Absolute continuity of the Floquet spectrum for a nonlinearly forced harmonic oscillator. Commun. Math. Phys.  \textbf{215}, 245-250 (2000)


\bibitem{GP19}
 Gr\'{e}bert, B.,  Paturel, E.: On reducibility of quantum harmonic oscillator on $\mathbb R^d$ with quasiperiodic in time potential. 
 Ann. Fac. Sci. Toulouse Math. \textbf{28}, 977-1014 (2019) 

\bibitem{GT2011}
  Gr\'{e}bert, B., Thomann,  L.: KAM for the quantum harmonic oscillator. Commun. Math. Phys. \textbf{307}, 383-427 (2011)

\bibitem{Kuk1993}
Kuksin, S. B.: Nearly integrable infinite-dimensional Hamiltonian systems. Lecture Notes in Mathematics \textbf{1556}, Springer-Verlag, Berlin (1993)

 \bibitem{LiangLuo2021}
Liang, Z., Luo, J.: Reducibility of 1-d quantum harmonic oscillator equation with  unbounded oscillation perturbations.  J. Differ. Equ.  \textbf{270}, 343-389 (2021)

\bibitem{LiangW2019}
Liang, Z. , Wang, Z.: Reducibility of quantum harmonic oscillator on $\mathbb R^d$ with differential and quasi-periodic in time potential. J. Differ. Equ.  \textbf{267},  3355-3395 (2019)
 
\bibitem{LW2021}
Liang, Z., Wang, Z.-Q.: Reducibility of quantum harmonic oscillator on $\R^d$ perturbed by q quasi-periodic potential with logarithmic decay.  submitted (2021)

\bibitem{LZZ2020}
Liang, Z., Zhao, Z., Zhou, Q.: 1-d quasi-periodic quantum harmonic oscillator with quadratic time-dependent perturbations: Reducibility and growth of Sobolev norms. 
J. Math. Pures Appl.  \textbf{146}, 158-182 (2021)

\bibitem{LY10}
  Liu, J., Yuan, X.: Spectrum for quantum duffing oscillator and small-divisor equation with large-variable coefficient. Commun. Pure Appl. Math. \textbf{63}, 1145-1172 (2010)
  
\bibitem{Mas2018}
Maspero, A.: Lower bounds on the growth of Sobolev norms in some linear time dependent Schr\"odinger equations. Math. Res. Lett. \textbf{26}, 1197-1215 (2019)

\bibitem{MR2017}
Maspero, A., Robert, D.: On time dependent Schr\"odinger equations: Global well-posedness and growth of Sobolev norms. J. Funct. Anal. \textbf{273}, 721-781 (2017).
  
\bibitem{Mon19}
Montalto, R.: A reducibility result for a class of linear wave equations on $\mathbb T^d$. Int. Math. Res. Not. \textbf{2019}, 1788-1862 (2019).

\bibitem{PrXu13}
Procesi, M., Xu, X.: Quasi-T\"oplitz functions in KAM theorem. SIAM J. Math. Anal. \textbf{45}, 2148-2181 (2013)

\bibitem{Shubin1987}
Shubin, M. A.: Pseudodifferential operators and spectral theory, \textbf{2nd} edition. Springer-Verlag, Berlin (2001)

 \bibitem{Th2020}
 Thomann, L.: Growth of Sobolev norms for linear Schr\"odinger operators. To appear in Pure Appl. Anal.  arXiv: 2006.02674 (2020)

\bibitem{Wang08}
Wang, W.-M.: Pure point spectrum of the Floquet Hamiltonian for the quantum harmonic oscillator under time quasi-periodic perturbations. Commun. Math. Phys. \textbf{277},
459-496 (2008)

\bibitem{WL2017} Wang, Z., Liang, Z.: Reducibility of 1D quantum harmonic oscillator perturbed by a quasiperiodic potential with logarithmic decay. Nonlinearity \textbf{30}, 1405-1448 (2017)



\end{thebibliography}
\end{document}